\documentclass[journal,twocolumn]{IEEEtran}%
\usepackage{amssymb}
\usepackage{amsfonts}
\usepackage{amsmath}
\usepackage{epic}
\usepackage{geometry}%
\usepackage{graphicx}
\newtheorem{theorem}{Theorem}

\newtheorem{corollary}{Corollary}

\newtheorem{definition}{Definition}
\newtheorem{example}{Example}

\newtheorem{lemma}{Lemma}

\newtheorem{remark}{Remark}

\begin{document}

\title{Volatility of Power Grids under Real-Time Pricing\vspace*{0.2in}}
\author{Mardavij Roozbehani,~\IEEEmembership{Member,~IEEE,}, Munther A
Dahleh,~\IEEEmembership{Member,~IEEE}, and Sanjoy K
Mitter,~\IEEEmembership{Member,~IEEE} \thanks{This work was supported by
Siemens AG under "Advanced Control Methods for Complex Networked Systems":
Siemens - MIT cooperation.}\thanks{The authors are with the Laboratory for
Information and Decision Systems (LIDS), Department of Electrical Engineering
and Computer Science, Massachusetts Institute of Technology, Cambridge, MA.
\newline E-mails: \{mardavij, dahleh, mitter\}@mit.edu.}\vspace*{-0.15in}}
\maketitle

\begin{abstract}
The paper proposes a framework for modeling and analysis of the dynamics of
supply, demand, and clearing prices in power system with real-time retail
pricing and information asymmetry. Real-time retail pricing is characterized
by passing on\textbf{\ }the real-time wholesale electricity prices to the end
consumers, and is shown to create a closed-loop feedback system between the
physical layer and the market layer of the power system. In the absence of a
carefully designed control law, such direct feedback between the two layers
could increase volatility and lower the system's robustness to uncertainty in
demand and generation. A new notion of \textit{generalized price-elasticity}
is introduced, and it is shown that price volatility can be characterized in
terms of the system's \textit{maximal relative price elasticity,} defined as
the maximal ratio of the generalized price-elasticity of consumers to that of
the producers. As this ratio increases, the system becomes more volatile, and
eventually, unstable. As new demand response technologies and distributed
storage increase the price-elasticity of demand, the architecture under
examination is likely to lead to increased volatility and possibly
instability. This highlights the need for assessing architecture
systematically and in advance, in order to optimally strike the trade-offs
between volatility, economic efficiency, and system reliability.

\end{abstract}

\markboth{LIDS Report, 2011}{Shell \MakeLowercase{\textit{et al.}}: Bare Demo of
IEEEtran.cls for Journals}

\begin{keywords}
Real-Time Pricing, Volatility, Lyapunov Analysis.
\end{keywords}

\section{Introduction}

\PARstart{T}{he} increasing demand for energy along with growing environmental
concerns have led to a national agenda for engineering a modern power grid
with the capacity to integrate the renewable energy resources at large scale.
In this paradigm shift, demand response and dynamic pricing are often promoted
as means of mitigating the uncertainties of the renewable generation and
improving the system's efficiency with respect to economic and environmental
metrics. The idea is to allow the consumers to react\textbf{--}in their own
monetary or environmental interest\textbf{--}to the wholesale market
conditions, possibly the real-time prices. However, this real-time or near
real-time coupling between supply and demand creates new challenges for power
system operation. The source of a most significant challenge is the
information asymmetry between the consumers and the system operators. Indeed,
real-time pricing under information asymmetry induces additional uncertainties
due to the uncertainty in consumer behavior, preferences, private valuation
for electricity, and consequently, unpredictable reactions to real-time prices.

The existing body of literature on dynamic pricing in communication or
transportation networks is extensive.\textbf{\ }See for instance
\cite{Kelly1998}, \cite{Chiang2007}, \cite{DynamicPricing} and the references
therein. However, the specific characteristics of power systems which can be
attributed to uncertainty in consumer behavior, the close coupling and
real-time interaction of economics and physics, and the reliability and
operational requirements that supply must match demand at all times raise very
unique challenges that need to be addressed.

Various forms of dynamic retail pricing of electricity have been advocated in
economic and engineering texts. In \cite{Borenstein2002}, Borenstein et. al.
study both the theoretical and the practical implications of
different\ dynamic pricing schemes such as \textit{Critical Peak Pricing},
\textit{Time-of-Use Pricing}, and \textit{Real-Time Pricing}. They argue in
favor of real-time pricing, characterized by passing on a price, that best
reflects the wholesale market prices, to the end consumers. They conclude that
real-time pricing delivers the most benefits in the sense of reducing the peak
demand and flattening the load curve. In \cite{Hogan}, Hogan identifies
dynamic pricing, particularly real-time pricing as a priority for
implementation of demand response in organized wholesale energy markets.
Similar conclusions are reached in a study conducted by Energy Futures
Australia \cite{Crosslley2008}.

The appeal of dynamic retail pricing is not limited to theoretical research
and academic studies, and real-world implementations are emerging at a rapid
pace. For instance, California's state's Public Utility Commission has enacted
a series of new energy regulations which set a deadline of 2011 for the state
utilities to propose a new \textit{dynamic pricing} rate structure,
specifically defined as an electric rate structure that reflects the actual
wholesale market conditions, such as critical peak pricing or real-time
pricing \cite{Calmess}. In this paper, we show that directly linking the
consumer prices to the wholesale market prices creates a close-loop feedback
system with the Locational Marginal Prices as the state variables. We observe
that such feedback mechanisms may increase volatility and decrease the
market's robustness to uncertainty in demand and generation. We introduce a
notion of \textit{generalized} price-elasticity, and show that price
volatility can be upperbounded by a function of the system's Maximal Relative
Price-Elasticity (MRPE), defined as the maximal ratio of the generalized
price-elasticity of consumers to the generalized price-elasticity of
producers. As this ratio increases, the system may become more volatile,
eventually becoming unstable as the MRPE exceeds one.

While the system can be stabilized and volatility can be reduced in many
different ways, e.g., via static or dynamic controllers regulating the
interaction of wholesale markets and retail consumers, different pricing
mechanisms pose different consequences on competing factors of interest such
as volatility, operational reliability, economic efficiency, and environmental
efficiency. The intended message is that the design of a real-time pricing
mechanism must take system stability issues into consideration, and that
successful design and implementation of such a mechanism entails careful
analysis of consumer behavior in response to price signals, and the trade-offs
between volatility, reliability, and economic or environmental efficiency.

Prior research relevant to stability of power markets has appeared in several
papers by\ Alvarado \cite{Alvarado1997}, \cite{Alvarado1999} on dynamic
modeling and stability, Watts and Alvarado \cite{Alvarado2004} on the
influence of future markets on price stability, and Nutaro and Protopopescu
\cite{Nutaro2009} on the impact of market clearing time and price signal delay
on power market stability. The model adopted in this paper differs from those
of \cite{Alvarado1997}, \cite{Alvarado1999}, \cite{Alvarado2004}, and
\cite{Nutaro2009} in that we analyze the global properties of the full
non-linear model as opposed to the first-order linear differential equations
examined in these papers. In addition, the price updates in our paper occur at
discrete time intervals and are an outcome of marginal cost pricing in the
wholesale market by an Independent System Operator (ISO), which is consistent
with the current practice in deregulated electricity markets. Furthermore,
beyond stability, we are interested in providing a characterization of the
impacts of uncertainty in consumer behavior on price volatility and the
system's robustness to uncertainties.

The organization of this paper is as follows. In Section \ref{sec:prem} we
present some preliminary concepts and definitions. In Section \ref{DSDM} we
present a mathematical model for the dynamic evolution of supply, demand, and
clearing prices under real-time pricing. Section \ref{sec:main} contains the
main theoretical contributions of this paper: characterizing volatility in
terms of the market's maximal relative elasticity and uncertainty in consumer
behavior. In Section \ref{sec:dis} we qualitatively discuss our results,
compare with some of the results in the literature, and point to some
important questions regarding the trade-offs arising due to uncertainty in
generation and quantifying the value of information. Numerical simulations are
presented in Section \ref{sec:sim}. Finally, we offer some closing remarks and
further directions for future research in Section \ref{sec:concl}.

\section{Preliminaries\label{sec:prem}}

\subsection{Notation}

The set of positive real numbers (integers) is denoted by $\mathbb{R}_{+}$
($\mathbb{Z}_{+}$)$,$ and nonnegative real numbers (integers) by
$\overline{\mathbb{R}}_{+}$ ($\overline{\mathbb{Z}}_{+}$). The class of
real-valued functions with a continuous $n$-th derivative on $X\subset
\mathbb{R}$ is denoted by $\mathcal{C}^{n}X.$ For a vector $v\in\mathbb{R}%
^{l},$ $v_{k}$ denotes the $k$-th element of $v,$ and $\left\Vert v\right\Vert
_{p}$ denotes the standard p-norm: $\left\Vert v\right\Vert _{p}%
\overset{\text{def}}{=}\left(
{\textstyle\sum\nolimits_{i=1}^{l}}
\left\vert v_{i}\right\vert ^{p}\right)  ^{1/p}$. Also, we will use
$\left\Vert v\right\Vert $ to denote any p-norm when there is no ambiguity.
The space of $\mathbb{R}^{l}$-valued functions $h:\mathbb{Z}\mapsto
\mathbb{R}^{l}$ of finite $p$-norm%
\[
\left\Vert h\right\Vert _{p}^{p}=%
{\displaystyle\sum\limits_{t=-\infty}^{\infty}}
\left\Vert h\left(  t\right)  \right\Vert _{p}^{p}=%
{\displaystyle\sum\limits_{t=-\infty}^{\infty}}
\sum\limits_{i=1}^{l}\left\vert h_{i}\left(  t\right)  \right\vert ^{p}%
\]
is denoted by $\ell_{p}\left(  \mathbb{Z}\right)  $ or simply $\ell_{p}$ when
there is no ambiguity. For a differentiable function $f:\mathbb{R}^{n}%
\mapsto\mathbb{R}^{m},$ we use $\dot{f}$ to denote the Jacobian matrix of $f.$
When $f$ is a scalar function of a single variable, $\dot{f}$ simply denotes
the derivative of $f$ with respect to its argument: $\dot{f}\left(  x\right)
=df\left(  x\right)  /dx.$ Since throughout the paper time is a discrete
variable, this notation would not be confused with derivative with respect to
time. Finally, for a measurable set $X\subset\mathbb{R},$ $\mu_{L}\left(
X\right)  $ is the Lebesgue measure of $X.$

\subsection{Basic Definitions}

\subsubsection{Volatility}

\begin{definition}
\label{SAvVDef}\textbf{Scaled} \textbf{Incremental Mean Volatility (IMV):}
Given a signal $h:\mathbb{Z}\mapsto\mathbb{R}^{l},$ and a function
$\rho:\mathbb{R}^{l}\mapsto\mathbb{R}^{m},$ the $\rho$\textit{-scaled
incremental mean volatility} measure of $h\left(  \cdot\right)  $ is defined
as%
\begin{equation}
\overline{\mathcal{V}}_{\rho}\left(  h\right)  =\lim_{T\rightarrow\infty}%
\frac{1}{T}%
{\displaystyle\sum\limits_{t=0}^{T}}
\left\Vert \rho\left(  h\left(  t+1\right)  \right)  -\rho\left(  h\left(
t\right)  \right)  \right\Vert \label{averagevol}%
\end{equation}
where, to simplify the notation, the dependence of the measure on the norm
used in (\ref{averagevol}) is dropped from the notation $\overline
{\mathcal{V}}_{\rho}\left(  h\right)  $.
\end{definition}

To quantify volatility for fast-decaying signals with zero or small scaled
IMV, e.g., state variables of a strictly stable autonomous system, we will use
the notion of scaled \textit{aggregate} volatility, defined as follows.

\begin{definition}
\label{SAgVDef}\textbf{Scaled\ Incremental\ Aggregate\ Volatility\ (IAV):}
Given a signal $h:\mathbb{Z}\mapsto\mathbb{R}^{l},$ and a function
$\rho:\mathbb{R}^{l}\mapsto\mathbb{R}^{m},$ the $\rho$\textit{-scaled
incremental aggregate volatility} measure of $h\left(  \cdot\right)  $ is
defined as%
\begin{equation}
\mathcal{V}_{\rho}\left(  h\right)  =%
{\displaystyle\sum\limits_{t=0}^{\infty}}
\left\Vert \rho\left(  h\left(  t+1\right)  \right)  -\rho\left(  h\left(
t\right)  \right)  \right\Vert .\label{aggregatevol}%
\end{equation}
In particular, we will be interested in the $\log$-scaled incremental
volatility as a metric for quantifying volatility of price, supply, or demand
in electricity markets.
\end{definition}

\begin{remark}
The notions of incremental volatility presented in Definitions \ref{SAvVDef}
and \ref{SAgVDef} accentuate the fast time scale, i.e., high frequency
characteristics of the signal of interest. Roughly speaking, the scaled IMV or
IAV are measures of the mean deviations of the signal from its \textit{moving
average}. In contrast, sample variance or CV (coefficient of variation,\ i.e.,
the ratio of standard deviation to mean) provide a measure of the mean
deviations of the signal from its average, without necessarily emphasizing the
high-frequency characteristics. A slowly-varying signal with a large dynamic
range may have a large sample variance or CV, but a small IMV, and thus will
be considered less volatile than a fast-varying signal with a large scaled
IMV. Since we are interested in studying the fast dynamics of spot prices in
electricity markets and the associated stability/reliability threats, the
scaled IMV and IAV as defined above are more appropriate measures of
volatility than variance or CV.
\end{remark}

\subsubsection{Stability}

The notion of stability used in this paper is the standard notion of
asymptotic stability. Consider the system
\begin{equation}
x\left(  t+1\right)  =\psi\left(  x\left(  t\right)  \right) \label{basicsys}%
\end{equation}
where $\psi\left(  \cdot\right)  $ is an arbitrary map from a domain
$X\subset\mathbb{R}^{n}$ to $\mathbb{R}^{n}.$ The equilibrium $\bar{x}\in X$
of (\ref{basicsys}) is \textit{stable in the sense of Lyapunov} if all
trajectories that start sufficiently close to $\bar{x}$ remain arbitrarily
close to it, i.e., for every $\varepsilon>0$,\ there exists $\delta>0$ such
that%
\[
\left\Vert x\left(  0\right)  -\bar{x}\right\Vert <\delta\Rightarrow\left\Vert
x\left(  t\right)  -\bar{x}\right\Vert <\varepsilon,\quad\forall t\geq0
\]
The equilibrium is \textit{globally asymptotically stable} if it is Lyapunov
stable and for all $x\left(  0\right)  \in X:\lim_{t\rightarrow\infty
}~x\left(  t\right)  =\bar{x}$.

\subsection{Market Structure \label{MarketSec}}

We begin with developing an electricity market model with three participants:
1. The suppliers, 2. The consumers, and 3. An Independent System Operator
(ISO). The suppliers and the consumers are price-taking, profit-maximizing
agents. The ISO is an independent, profit-neutral player in charge of clearing
the market, that is, matching supply and demand subject to the network
constraints with the objective of maximizing the social welfare. Below, we
describe the characteristics of the players in more detail.

\subsubsection{The Consumers and the Producers}

Let $D=\left\{  1,\ldots,n_{s}\right\}  $ and $S=\left\{  1,\ldots
,n_{s}\right\}  $ denote the sets of consumers and producers respectively.
Each consumer $j\in D$ is associated with a value function $v_{j}%
:\overline{\mathbb{R}}_{+}\mapsto\mathbb{R},$ which can be thought of as the
monetary value that consumer $j$ derives from consuming $x$ units of the
resource, electricity in this case. Similarly, each producer $i\in S,$ is
associated with a function $c_{i}:\overline{\mathbb{R}}_{+}\mapsto
\overline{\mathbb{R}}_{+}$ representing the monetary cost of production of the resource.

\textit{Assumption I:} For all $i\in S,$ the cost functions $c_{i}\left(
\cdot\right)  $ are in $\mathcal{C}^{2}(0,\infty),$ strictly increasing, and
strictly convex. For all $j\in D,$ the value functions $v_{j}\left(
\cdot\right)  $ are in $\mathcal{C}^{2}(0,\infty)$, strictly increasing, and
strictly concave.

Let $d_{j}:\mathbb{R}_{+}\mapsto\overline{\mathbb{R}}_{+},$~$j\in D,$ and
$s_{i}:\mathbb{R}_{+}\mapsto\overline{\mathbb{R}}_{+},~i\in S$ be demand and
supply functions mapping price to consumption and production, respectively. In
the framework adopted in this paper, the producers and consumers are
price-taking, utility-maximizing agents. Therefore, letting $\lambda$ be the
price per unit of electricity, we have%
\begin{align}
d_{j}\left(  \lambda\right)   &  =\arg\underset{x\in\overline{\mathbb{R}}_{+}%
}{\max}~~v_{j}\left(  x\right)  -\lambda x,\text{\qquad}j\in D,\label{D P}\\
&  =\max\left\{  0,\left\{  x~|~\dot{v}_{j}\left(  x\right)  =\lambda\right\}
\right\} \nonumber
\end{align}
and%
\begin{align}
s_{i}\left(  \lambda\right)   &  =\arg\underset{x\in\overline{\mathbb{R}}_{+}%
}{\max}~~\lambda x-c_{i}\left(  x\right)  ,\text{\qquad}i\in S.\label{S P}\\
&  =\max\left\{  0,\left\{  x~|~\dot{c}_{i}\left(  x\right)  =\lambda\right\}
\right\} \nonumber
\end{align}
For the sake of convenience in notation and in order to avoid unnecessary
technicalities, unless stated otherwise, we will assume in the remainder of
this paper that $d_{j}\left(  \lambda\right)  =\dot{v}_{j}^{-1}\left(
\lambda\right)  $ is the demand function, and $s_{i}\left(  \lambda\right)
=\dot{c}_{i}^{-1}\left(  \lambda\right)  $ is the supply function. This can be
mathematically justified by assuming that $\dot{v}\left(  0\right)  =\infty,$
and $\dot{c}\left(  0\right)  =0$, or that $\lambda\in\left[  \dot{c}\left(
0\right)  ,\dot{v}\left(  0\right)  \right]  .$

\vspace*{0.1in}

\begin{definition}
The social welfare $\mathcal{S}$ is the aggregate benefit of the producers and
the consumers:\vspace*{0.1in}%
\[
\mathcal{S}=\sum\limits_{j\in D}\left(  v_{j}\left(  d_{j}\right)
-\lambda_{j}d_{j}\right)  -\sum\limits_{i\in S}\left(  \lambda_{i}s_{i}%
-c_{i}\left(  s_{i}\right)  \right)
\]
If $\lambda_{i}=\lambda_{j}=\lambda,$ $\forall i,j,$ we say that $\lambda$ is
a uniform market clearing price, and in this case, we have:\vspace*{0.1in}%
\[
\mathcal{S}=\sum\limits_{j\in D}v_{j}\left(  d_{j}\right)  -\sum_{i\in S}%
c_{i}\left(  s_{i}\right)
\]

\end{definition}

\paragraph{Heterogeneous Consumers with Uncertain Value Functions}

We will consider two models of heterogenous consumers with time-varying value
functions to represent the uncertainty in consumer behavior.

\paragraph*{\textbf{--} \textbf{Multiplicative Perturbation Model}}

The uncertainty in consumer's value function is modeled as%
\begin{equation}
\tilde{v}_{j}\left(  x,t\right)  =\alpha_{j}\left(  t\right)  v_{o}\left(
\frac{x}{\alpha_{j}\left(  t\right)  }\right)  ,\text{\qquad}j\in
D,\label{TVV2}%
\end{equation}
where $v_{o}:\overline{\mathbb{R}}_{+}\mapsto\mathbb{R}$ is a nominal value
function and $\alpha_{j}:\overline{\mathbb{Z}}_{+}\mapsto\mathbb{R}_{+}$ is an
exogenous signal or disturbance. Given a price $\lambda\left(  t\right)  >0$,
under the multiplicative perturbation model (\ref{TVV2}) we have%
\begin{equation}
d_{j}\left(  \lambda,t\right)  =\alpha_{j}\left(  t\right)  \dot{v}_{o}%
^{-1}\left(  \lambda\left(  t\right)  \right) \label{TVDF2}%
\end{equation}
Thus, the same price $\lambda$ may induce different consumptions at different
times, depending on the type and composition of the load.

\paragraph*{\textbf{--} \textbf{Additive Perturbation Model}}

The uncertainty in consumer's value function is modeled as%
\begin{equation}
\tilde{v}_{j}\left(  x,t\right)  =v_{o}\left(  x-u_{j}\left(  t\right)
\right)  ,\text{\qquad}j\in D,\label{TVV1}%
\end{equation}
where $u_{j}:\overline{\mathbb{Z}}_{+}\mapsto\mathbb{R}_{+}$ is exogenous.
Thus, given a price $\lambda\left(  t\right)  >0$, under the additive
perturbation model (\ref{TVV1}), the demand function is%
\begin{equation}
d_{j}\left(  \lambda,t\right)  =u_{j}\left(  t\right)  +\dot{v}_{o}%
^{-1}\left(  \lambda\left(  t\right)  \right) \label{TVDF1}%
\end{equation}

\paragraph*{\textbf{--} \textbf{Aggregation of Several Consumers}}

The aggregate response of several consumers (or producers)\ to a price signal
may be modeled as the response of a single representative agent, although
explicit formula for the utility of the representative agent may sometimes be
too complicated to find \cite{RoozbehaniCDC2010, Hartley1997}. For the case of
$N$ identical consumers with value functions $v_{j}=v_{o},$\quad$j\in D,$ it
can be verified that the aggregate demand is equivalent to the demand of a
representative consumer with value function \cite{RoozbehaniCDC2010}:%
\begin{equation}
v\left(  x\right)  =Nv_{o}\left(  \frac{x}{N}\right) \label{AggV}%
\end{equation}
Suppose now, that the consumer behavior can be modeled via (\ref{TVV2}%
)--(\ref{TVDF2}). Let
\[
\bar{\alpha}\left(  t\right)  =\sum\nolimits_{j=1}^{N}\alpha_{j}\left(
t\right)  ,
\]
and suppose that there exists a nominal value $\bar{\alpha}_{0}$, such that
\[
\bar{\alpha}\left(  t\right)  =\bar{\alpha}_{0}+\Delta\bar{\alpha}\left(
t\right)  =\bar{\alpha}_{0}\left(  1+\delta\left(  t\right)  \right)
\]
where $\delta\left(  t\right)  =\Delta\bar{\alpha}\left(  t\right)
/\bar{\alpha}_{0}$ satisfies $\left\vert \delta\left(  t\right)  \right\vert
<1$. Define $v\left(  x\right)  =\bar{\alpha}_{0}v_{o}\left(  x/\bar{\alpha
}_{0}\right)  .$ It can be then verified that the aggregate demand can be
modeled as the response of a representative agent with value function
\begin{align}
\tilde{v}\left(  x,t\right)   & =\bar{\alpha}\left(  t\right)  v_{o}\left(
\frac{x}{\bar{\alpha}\left(  t\right)  }\right) \nonumber\\
& =\left(  \bar{\alpha}_{0}+\Delta\bar{\alpha}\left(  t\right)  \right)
v_{o}\left(  \frac{x}{\bar{\alpha}_{0}+\Delta\bar{\alpha}\left(  t\right)
}\right) \nonumber\\
& =\left(  1+\delta\left(  t\right)  \right)  v\left(  \frac{x}{1+\delta
\left(  t\right)  }\right) \label{TVV2a}%
\end{align}
The aggregate response is then given by
\begin{equation}
d\left(  \lambda\left(  t\right)  ,t\right)  =\left(  1+\delta\left(
t\right)  \right)  \dot{v}^{-1}\left(  \lambda\left(  t\right)  \right)
.\label{TVDF2a}%
\end{equation}

Similarly, under the additive perturbation model the aggregate behavior can be
represented by%
\begin{align}
\tilde{v}\left(  x,t\right)   & =v\left(  x-u\left(  t\right)  \right)
\label{TVV1a}\\
d\left(  \lambda\left(  t\right)  ,t\right)   & =u\left(  t\right)  +\dot
{v}^{-1}\left(  \lambda\left(  t\right)  \right) \label{TVDF1a}%
\end{align}
where $v\left(  \cdot\right)  $ is given by (\ref{AggV}) and $u\left(
t\right)  =\sum u_{j}\left(  t\right)  .$ The interpretation of (\ref{TVV1a})
and (\ref{TVDF1a}) is that at any given time $t,$ the demand comprises of an
inelastic component $u\left(  t\right)  $ which is exogenous$,$ and an elastic
component $\dot{v}^{-1}\left(  \lambda\left(  t\right)  \right)  $. Another
interpretation is that $\dot{v}^{-1}\left(  \lambda\left(  t\right)  \right)
$ represents the demand of those consumers who are subject to real-time
pricing, and $u\left(  t\right)  $ represents the demand of the
non-participating consumers.

\subsubsection{The Independent System Operator (ISO)}

The ISO is a non-for-profit entity whose primary function\ is to optimally
match supply and demand\textbf{\ }subject to network and operational
constraints. The constraints include power flow constraints (Kirchhoff's
laws), transmission line constraints, generator capacity constraints, local
and system-wide reserve capacity requirements and possibly some other
constraints specific to the ISO \cite{ISO-NE}, \cite{PJM-ISO}, \cite{Litvinov}%
. For real-time market operation, the constraints are linearized near the
steady-state operating point and the ISO optimization problem is reduced to a
convex\textbf{--}typically linear\textbf{--}optimization often referred to as
the \textit{Economic Dispatch Problem (EDP), or the Optimal Power Flow
Problem}. A set of Locational Marginal Prices (LMP)\ emerge as the dual
variables corresponding to the nodal power balance constraints. These prices
vary from location to location as they represent the marginal cost of
supplying electricity to a particular location. We refer the interested reader
to \cite{SCTR1998}, \cite{Litvinov}, \cite{PJM-ISO}, and \cite{Roozbehani2010}
for more details. However, we emphasize that the spatial variation in the LMPs
is a consequence of congestion in the transmission lines. When there is
sufficient transmission capacity in the network, a uniform\ price will
materialize for the entire system. With this observation in sight, and in
order to develop tractable models that effectively highlight the impacts of
the behavior of producers and consumers\textbf{--}quantified through their
cost and value functions\textbf{--}on system stability and price volatility,
we will make the following\ assumptions:

\begin{enumerate}
\item Resistive losses are negligible.

\item The line capacities are high enough, so, congestion will not occur.

\item There are no generator capacity constraints.

\item The system always has sufficient reserve capacity and the marginal cost
of reserve is the same as the marginal cost of generation.
\end{enumerate}

Under the first two assumptions, the network parameters become irrelevant in
the supply-demand optimal matching problem. The third and fourth assumptions
are made in the interest of keeping the development in this paper focused.
They could, otherwise, be relaxed at the expense of a somewhat more involved
technical analysis. The last assumption also implies that we do not
differentiate between actual generation and reserve. A thorough investigation
of the effects of network constraints and reserve capacity markets, whether
they are stabilizing or destabilizing, does not fall within the scope of this
paper. The interested readers may consult \cite{Alvarado1999, Roozbehani2010}
for an analysis of dynamic pricing in electricity networks with transmission
line and generator capacity constraints.

Under the above assumptions, the following problem characterizes the ISO's
optimization problem:
\begin{equation}%
\begin{array}
[c]{ccc}%
\max &  &
{\displaystyle\sum\limits_{j\in D}}
v_{j}(d_{j})-%
{\displaystyle\sum\limits_{i\in S}}
c_{i}(s_{i})\vspace{0.15in}\\
\text{s.t.} &  &
{\displaystyle\sum\limits_{j\in D}}
d_{j}=%
{\displaystyle\sum\limits_{i\in S}}
s_{i}%
\end{array}
\label{Market P}%
\end{equation}

The following lemma which is adopted from \cite{Kelly1997}, provides the
justification for defining the LMPs as the Lagrangian multipliers
corresponding to the balance constraint(s).

\begin{lemma}
\label{Kelly}Let $d^{\ast}=\left[  d_{1}^{\ast},\cdots,d_{n_{d}}^{\ast
}\right]  ,$ and $s^{\ast}=\left[  s_{1}^{\ast},\cdots,s_{n_{s}}^{\ast
}\right]  $ where $d_{j}^{\ast},$~$~j\in D$ and $s_{i}^{\ast},$~$~i\in S$,
solve (\ref{Market P}). There exists a price $\lambda^{\ast}\in\left(
0,\infty\right)  ,$ such that $d^{\ast}$ and $s^{\ast}$ solve (\ref{D P}) and
(\ref{S P}). Furthermore, $\lambda^{\ast}$ is the Lagrangian multiplier
corresponding to the balance constraint in (\ref{Market P}).
\end{lemma}

\begin{proof}
The proof is based on Lagrangian duality and is omitted for brevity. The proof
in \cite{Kelly1997} would be applicable here with some minor adjustments.
\end{proof}

The implication of Lemma \ref{Kelly} is that by defining\ the market price to
be the vector of Lagrangian multiplier corresponding to the balance
constraints, the system operator creates a competitive environment in which,
the collective selfish behavior of the participants results in a system-wide
optimal condition.

\paragraph{Real-Time System Operation and Market Clearing}

Consider the case of real-time market operation and assume that
price-sensitive retail consumers do not bid in the real-time market. In other
words, they do not provide their value functions to the system operator (or
any intermediary entity in charge of real-time pricing). Though, they may
adjust their consumption in response to a price signal, which is assumed in
this paper, to be the wholesale market clearing price. In this case, the
demand is assumed to be inelastic over each \textit{short} pricing interval,
and supply is matched to demand. Therefore, (\ref{Market P}) reduces to
meeting the fixed demand at minimum cost:%
\begin{equation}%
\begin{array}
[c]{ccl}%
\min &  &
{\displaystyle\sum\limits_{i\in S}}
c_{i}(s_{i})\vspace{0.15in}\\
\text{s.t.} &  &
{\displaystyle\sum\limits_{i\in S}}
s_{i}=%
{\displaystyle\sum\limits_{j\in D}}
\hat{d}_{j}%
\end{array}
\label{Market P fixed}%
\end{equation}
where $\hat{d}_{j}$ is the predicted demand of consumer $j$ for the next time
period. We assume that the system operator solves (\ref{Market P fixed}) and
sets the price to the marginal cost of production at the minimum cost
solution. The discrepancy between scheduled generation (which is equal to the
predicted demand) and actual demand is compensated through reserves with the
same marginal costs. Thus, we will not include reserve parameters and
equations explicitly in the model. More details regarding a dynamic extensions
of this model are presented in the next Subsection.

\section{Dynamic Models of Supply-Demand under Information
Asymmetry\label{DSDM}}

In this section, we develop dynamical system models for the interaction of
wholesale supply and retail demand in electricity markets with information
asymmetry. In this context, \textquotedblleft asymmetry of
information\textquotedblright\ refers to the architecture of the information
layer of the market, in which, the market operator has full information about
the cost of supplying the resource (e.g., through the offers of the
producers), but has no information about \textit{valuation} of the resource by
the demand side.
\begin{figure}
[ptb]
\begin{center}
\includegraphics[
height=1.9in,
width=2.2in
]%
{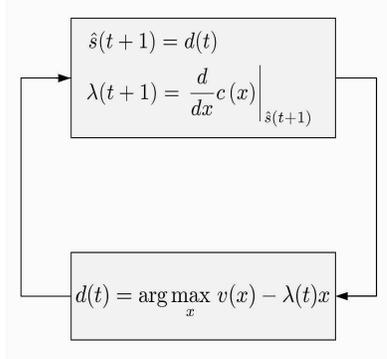}%
\caption{Exant\'{e} Priced Supply/Demand Feedback}%
\label{Feedback 1}%
\end{center}
\end{figure}

The real-time market is cleared at discrete time intervals and the prices are
calculated and announced for each interval\footnote{In most regions of the
United States, such as New England, California, or PJM, the real-time market
is operated in five-minute intervals.}. The practice of defining the clearing
price corresponding to each pricing interval based on the predicted demand at
the beginning of the interval\ is called exant\'{e} pricing. As opposed to
this, ex-post pricing refers to the practice of defining the clearing price
for each pricing interval based on the materialized consumption at the end of
the interval. In ex-post pricing the demand is subject to some price
uncertainty as the actual price will be revealed after consumption has
materialized. In exant\'{e} pricing without ex-post adjustments, the
entity\ in charge of real-time pricing faces the price uncertainty\footnote{In
this paper we combine the role of the ISO\ and that of an entity in charge of
real-time pricing in the retail market. Whether in practice this will be the
case or not, has no influence on the intended message and the results that we
deliver.}, as it will have to reimburse the generators based on the actual
marginal cost of production, while it can charge the demand only based on the
exant\'{e} price. We will present dynamic market models for both pricing
schemes. These models are consistent with the current practice of marginal
cost pricing in wholesale electricity markets, with the additional feature
that the retail consumers adjust their usage based on the real-time wholesale
market price.

\subsection{Price Dynamics under Exant\'{e} Pricing}

Let $\lambda\left(  t\right)  $ denote the exant\'{e} price corresponding to
the consumption of one unit of electricity in the time interval $\left[
t,t+1\right]  .$ Let $d\left(  t\right)  =%
{\textstyle\sum\nolimits_{j\in D}}
d_{j}\left(  t\right)  $ be the actual aggregate consumption during this
interval:
\begin{equation}
d\left(  t\right)  =\sum\limits_{j\in D}d_{j}\left(  t\right)  =\sum
\limits_{j\in D}v_{j}^{-1}\left(  \lambda\left(  t\right)  \right)
.\label{OptDemandInd}%
\end{equation}
Since $v_{j}\left(  \cdot\right)  $ is known only to consumer $j$, at time $t,
$ only an estimate of $d\left(  t\right)  $ is available to the ISO$,$ based
on which, the price $\lambda\left(  t\right)  $ is calculated$.$ The price
$\lambda\left(  t\right)  $ is therefore, the marginal cost of predicted
supply to meet the predicted demand for the time interval $[t,t+1]$. We assume
that the ISO's predicted demand/supply for each time interval ahead is based
on the actual demand at the previous intervals: $\hat{s}_{t+1}=\hat{d}\left(
t+1\right)  =\phi\left(  d\left(  t\right)  ,\cdots,d\left(  t-T\right)
\right)  ,$ $T\in\mathbb{Z}.$ The following equations describe the dynamics of
the market:%
\begin{align}
\sum\limits_{i\in S}\dot{c}_{i}^{-1}\left(  \lambda\left(  t+1\right)
\right)   & =\hat{s}\left(  t+1\right)  =\hat{d}\left(  t+1\right)
\label{MinCost}\\
\hat{d}\left(  t+1\right)   & =\phi\left(  d\left(  t\right)  ,\cdots,d\left(
t-T\right)  \right) \label{PredStep}\\[0.1in]
\sum\limits_{j\in D}\dot{v}_{j}^{-1}\left(  \lambda\left(  t-k\right)
\right)   & =d\left(  t-k\right)  ,\text{\qquad}\forall k\leq
T\label{OptDemand}%
\end{align}
where (\ref{OptDemand}) follows from (\ref{OptDemandInd}), and $\lambda\left(
t+1\right)  $ in (\ref{MinCost}) is the Lagrangian multiplier associated with
the balance constraint in optimization problem (\ref{Market P fixed})
solved\ at time $t+1,$ i.e., with $%
{\textstyle\sum\nolimits_{j\in D}}
\hat{d}_{j}=\hat{d}\left(  t+1\right)  .$

The prediction step (\ref{PredStep}) may be carried through by resorting to
linear auto-regressive models, in which case, we will have:%
\begin{equation}
\phi\left(  d\left(  t\right)  ,\cdots,d\left(  t-T\right)  \right)
=\sum\limits_{k=0}^{T}\alpha_{k}d\left(  t-k\right)  ,\text{\qquad}\alpha
_{k}\in\mathbb{R}{\normalsize .}\label{LARMAX}%
\end{equation}
When $\phi\left(  \cdot\right)  $ is of the form (\ref{LARMAX}), equations
(\ref{MinCost})$-$(\ref{OptDemand}) result in:%
\begin{equation}
\sum\limits_{i\in S}\dot{c}_{i}^{-1}\left(  \lambda\left(  t+1\right)
\right)  =\sum\limits_{k=0}^{T}\alpha_{k}\sum\limits_{j\in D}\dot{v}_{j}%
^{-1}\left(  \lambda\left(  t-k\right)  \right) \label{AX Price D}%
\end{equation}
\ Some of the popular forecasting models currently in use by the system
operators are variations of the \textit{persistence model} which corresponds
to the special case where the predicted demand for the next time step is
assumed to be equal to the demand at the previous time step, i.e.,
$\phi\left(  d\left(  t\right)  ,\cdots,d\left(  t-T\right)  \right)
=d\left(  t\right)  $. In this case, equations (\ref{MinCost})$-$%
(\ref{OptDemand}) result in:%
\begin{equation}
\sum\limits_{i\in S}\dot{c}_{i}^{-1}\left(  \lambda\left(  t+1\right)
\right)  =\sum\limits_{j\in D}\dot{v}_{j}^{-1}\left(  \lambda\left(  t\right)
\right)  .\label{Price D}%
\end{equation}
If all the producers can be aggregated into one representative producer agent
with a convex cost function $c\left(  \cdot\right)  ,$ and all the consumers
can be aggregated into one representative consumer agent with a concave value
function $v\left(  \cdot\right)  ,$ then (\ref{AX Price D}) and (\ref{Price D}%
) reduce, respectively, to :%
\begin{equation}
\lambda\left(  t+1\right)  =\dot{c}\left(  \sum\nolimits_{k=0}^{T}\alpha
_{k}\dot{v}^{-1}\left(  \lambda\left(  t-k\right)  \right)  \right)
\label{AX Price D Simple}%
\end{equation}
and%
\begin{equation}
\lambda\left(  t+1\right)  =\dot{c}\left(  \dot{v}^{-1}\left(  \lambda\left(
t\right)  \right)  \right)  .\label{Price D Simple}%
\end{equation}
\qquad More details on the construction of the representative agent mode can
be found in \cite{RoozbehaniCDC2010}.

\subsection{Price Dynamics under Ex-post Pricing}%

\begin{figure}
[ptb]
\begin{center}
\includegraphics[
height=1.9in,
width=2.2in
]%
{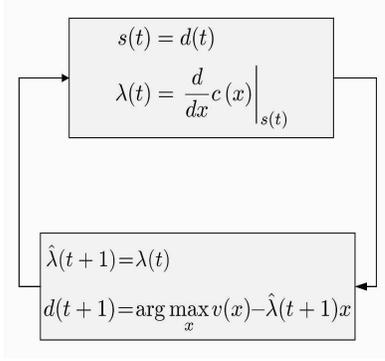}%
\caption{Ex-post Priced Supply/Demand Feedback}%
\label{Feedback 2}%
\end{center}
\end{figure}

Under ex-post pricing, the price charged for consumption of one unit of
electricity during the interval $\left[  t,t+1\right]  $ is declared at time
$t+1$, when the total consumption has materialized. In order to decide on the
amount to consume during each time interval ahead, a prediction of the ex-post
price is needed. We assume that $\hat{\lambda}_{j}\left(  t+1\right)  , $
consumer $j$'s predicted price, is a function\ of the ex-post prices of the
previous intervals. Therefore,%
\begin{align}
\hat{\lambda}_{j}\left(  t+1\right)   & =\phi_{j}\left(  \lambda\left(
t\right)  ,\cdots,\lambda\left(  t-T\right)  \right) \label{PricePredStep}%
\\[0.05in]
d\left(  t+1\right)   & =\sum\limits_{j\in D}\dot{v}_{j}^{-1}\left(
\hat{\lambda}^{j}\left(  t+1\right)  \right) \label{ExPOptDemand}\\
\sum\limits_{i\in S}\dot{c}_{i}^{-1}\left(  \lambda\left(  t+1\right)
\right)   & =d\left(  t+1\right) \label{ExPMinCost}%
\end{align}
By combining (\ref{PricePredStep})$-$(\ref{ExPMinCost}) we obtain:%
\begin{equation}
\sum\limits_{i\in S}\dot{c}_{i}^{-1}\left(  \lambda\left(  t+1\right)
\right)  =\sum\limits_{j\in D}\dot{v}_{j}^{-1}\left(  \phi_{j}\left(
\lambda\left(  t\right)  ,\cdots,\lambda\left(  t-T\right)  \right)  \right)
\label{ExPCmplxPriceD}%
\end{equation}
It is observed that when the consumers use the persistence model for
predicting future prices, i.e., when $\phi_{j}\left(  \lambda\left(  t\right)
,\cdots,\lambda\left(  t-T\right)  \right)  =\lambda\left(  t\right)  ,$
$\forall j,$ then the price dynamics (\ref{ExPCmplxPriceD}) becomes identical
to the case with exant\'{e} pricing (\ref{Price D}), with the difference that
the price uncertainty and the associated risks are bore by the consumer. In
general, however, the price dynamics would depend on how each individual
consumer predicts the ex-post price. This additional layer of dependency on
consumer behavior suggests that more complicated market outcomes with
multiple, possibly inefficient equilibria could materialize in ex-post-priced
retail markets.

\begin{remark}
Equation (\ref{ExPMinCost}) assumes that the generators were dispatched
optimally, which is ideal but unlikely in practice. In this paper, we do not
model the intricacies arising from the discrepancy between exant\'{e} dispatch
(which is the actual dispatch schedule based on prediction, and hence, not
necessarily optimal) and ex-post dispatch (which characterizes how the
generators should have been ideally dispatched). Although the settlement of
these discrepancies is important in practice, such details can be safely
ignored without affecting the core of our framework.
\end{remark}

\begin{remark}
It is also possible to consider dynamic models arising from exant\'{e} pricing
complemented with ex-post adjustments, see for instance \cite{MvichEEM11}.
\end{remark}

\subsection{Demand Dynamics under Exant\'{e} or Ex-post Pricing}

We could alternatively write dynamical system equations for the evolution\ of
demand. Under exant\'{e} pricing we will have:
\begin{align}
\dot{v}_{j}\left(  d_{j}\left(  t+1\right)  \right)   &  =\dot{c}_{i}\left(
s_{i}\left(  t\right)  \right)  \hspace{0.15in}\forall i\in S,~j\in
D\label{DDynExante1}\\[0.1in]
\sum\limits_{i\in S}s_{i}\left(  t\right)   &  =\phi\left(  \sum
\nolimits_{j\in D}d_{j}\left(  t\right)  ,\cdots,\sum\nolimits_{j\in D}%
d_{j}\left(  t-T\right)  \right)  ,\label{DDynExante2}%
\end{align}
whereas, under ex-post pricing we will have:%
\begin{align}
\dot{v}_{j}\left(  d_{j}\left(  t+1\right)  \right)   &  =\phi_{j}\left(
\dot{c}_{i}\left(  s_{i}\left(  t\right)  \right)  ,\cdots,\dot{c}_{i}\left(
s_{i}\left(  t-T\right)  \right)  \right) \label{DDynExpost1}\\[0.1in]
\sum\limits_{i\in S}s_{i}\left(  t\right)   &  =\sum\limits_{j\in D}%
d^{j}\left(  t\right)  .\label{DDynExpost2}%
\end{align}
Assuming representative agent models, (\ref{DDynExante1})$-$(\ref{DDynExante2}%
) and (\ref{DDynExpost1})$-$(\ref{DDynExpost2}) reduce, respectively, to%
\begin{equation}
\dot{v}\left(  d\left(  t+1\right)  \right)  =\dot{c}\left(  \phi\left(
d\left(  t\right)  ,\cdots,d\left(  t-T\right)  \right)  \right)
\label{TheEq3}%
\end{equation}
and
\begin{equation}
\dot{v}\left(  d\left(  t+1\right)  \right)  =\phi\left(  \dot{c}\left(
d\left(  t\right)  \right)  ,\cdots,\dot{c}\left(  d\left(  t-T\right)
\right)  \right)  .\label{TheEq4}%
\end{equation}
Under the persistence model for prediction we have:%
\begin{equation}
\dot{v}\left(  d\left(  t+1\right)  \right)  =\dot{c}\left(  d\left(
t\right)  \right) \label{Demand D}%
\end{equation}
In the sequel, we will develop a theoretical framework that is convenient for
analysis of dynamical systems described by\ implicit equations. Such systems
arise in many applications which incorporate real-time optimization in a
feedback loop, several instances of which were developed in this section. As
we will see, this framework is extremely useful for studying the dynamics of
electricity markets, robustness to disturbances, price stability, and
volatility under real-time pricing.

\section{Theoretical Framework\label{sec:main}}

\subsection{Stability Analysis}

In this section we present several stability criteria based on Lyapunov
techniques and examine stability properties of the clearing price dynamics
formulated in Section \ref{DSDM}.

\begin{theorem}
\label{Main}Let $\mathcal{S}$ be a discrete-time dynamical system described by
the state-space equation%
\begin{equation}%
\begin{array}
[c]{ccrll}%
\mathcal{S} & : & x\left(  t+1\right)  \hspace{-0.08in} & = & \hspace
{-0.08in}\psi\left(  x\left(  t\right)  \right)  \vspace*{0.1in}\\
&  & x_{0}\hspace{-0.08in} & \in & \hspace{-0.08in}X_{0}\subset\mathbb{R}_{+}%
\end{array}
\label{scalar DS}%
\end{equation}
for some function $\psi:\mathbb{R}_{+}\mapsto\mathbb{R}_{+}$. Then,
$\mathcal{S}$ is stable if there exists a pair of continuously differentiable
functions $f,g:\mathbb{R}_{+}\mapsto\mathbb{R}_{+}$ satisfying%
\begin{equation}
g\left(  x\left(  t+1\right)  \right)  =f\left(  x\left(  t\right)  \right)
\label{gp=f}%
\end{equation}
and%
\begin{align}
\hspace{-0.04in}\left(  \text{i}\right)   &  :\hspace{0.2in}\theta^{\ast}%
=\inf\left\{  \theta~|~\left\vert \dot{f}\left(  x\right)  \right\vert
\leq\theta\left\vert \dot{g}\left(  x\right)  \right\vert ,\hspace
{0.1in}\forall x\right\}  \leq1\label{fdlgd}\\[0.12in]
\hspace{-0.04in}\left(  \text{ii}\right)   &  :\hspace{0.2in}\mu
_{L}(\{x~|~\dot{f}\left(  x\right)  =\dot{g}\left(  x\right)
\})=0\label{fdlg4}%
\end{align}
and either:%
\begin{equation}
\hspace{0.17in}(\text{iii}):\hspace{0.16in}\dot{g}\left(  x\right)
\geq0,\text{ }\forall x,\text{ and }\lim\limits_{x\rightarrow\infty}\left\{
f\left(  x\right)  -g\left(  x\right)  \right\}  <0\label{fdlgd2}%
\end{equation}
or%
\begin{equation}
\hspace{0.17in}(\text{iii})^{\prime}\hspace{-0.02in}:\hspace{0.14in}\dot
{g}\left(  x\right)  \leq0,\text{ }\forall x,\text{ and }\lim
\limits_{x\rightarrow\infty}\left\{  f\left(  x\right)  -g\left(  x\right)
\right\}  >0\label{fdlgd3}%
\end{equation}
\vspace*{0.1in}
\end{theorem}

Before we proceed with proving Theorem \ref{Main}, we present the following
lemma, which will be used several times in this paper.

\begin{lemma}
\label{MainLemma}Let $X$ be a subset of $\mathbb{R}$. Suppose that there
exists a continuously differentiable function $f:X\mapsto\mathbb{R},$ a
continuously differentiable monotonic function $g:X\mapsto\mathbb{R},$ and a
constant $\theta\in\lbrack0,\infty)$ satisfying%
\begin{equation}
\left\vert \dot{f}\left(  x\right)  \right\vert \leq\theta\left\vert \dot
{g}\left(  x\right)  \right\vert ,\text{\quad}\forall x\in X\label{LemmaCond}%
\end{equation}
Then
\begin{equation}
\hspace*{-0.25in}\left\vert f\left(  x\right)  -f\left(  y\right)  \right\vert
\leq\theta\left\vert g\left(  x\right)  -g\left(  y\right)  \right\vert
,\text{\quad}\forall x,y\hspace*{-0.03in}\in\hspace*{-0.03in}X\label{LemmaEq1}%
\end{equation}
Furthermore, if (\ref{fdlg4}) is satisfied, then
\begin{equation}
\hspace*{-0.25in}\left\vert f\left(  x\right)  -f\left(  y\right)  \right\vert
<\left\vert g\left(  x\right)  -g\left(  y\right)  \right\vert ,\text{\quad
}\forall x,y\hspace*{-0.03in}\in\hspace*{-0.03in}X,\text{ }x\neq
y\label{LemmaEq2}%
\end{equation}

\end{lemma}

\vspace*{0.1in}

\begin{proof}
We have\vspace*{-2pt}%
\begin{align}
\hspace*{-0.25in}\forall x,y\hspace*{-0.03in}  & \in\hspace*{-0.03in}X,\text{
}x\neq y:\nonumber\\[0.05in]
\hspace*{-0.25in}\left\vert f\left(  x\right)  -f\left(  y\right)  \right\vert
\hspace*{-0.03in}  & \leq\hspace*{-0.03in}\left\vert \int_{y}^{x}\left\vert
\dot{f}\left(  \tau\right)  \right\vert d\tau\right\vert \nonumber\\[0.05in]
\hspace*{-0.25in}\hspace*{-0.03in}  & \leq\hspace*{-0.03in}\theta\left\vert
\int_{y}^{x}\left\vert \dot{g}\left(  \tau\right)  \right\vert d\tau
\right\vert =\theta\left\vert g\left(  x\right)  -g\left(  y\right)
\right\vert \label{eq3}%
\end{align}
where the inequality in (\ref{eq3}) follows from (\ref{LemmaCond}) and the
subsequent equality follows from (\ref{fdlgd2}). Proof of (\ref{LemmaEq2}) is
similar, except that under the assumptions of the lemma, the non-strict
inequality in (\ref{eq3}) can be replaced with a strict inequality.
\end{proof}

\vspace*{0.1in}

We will now present the proof of Theorem \ref{Main}.

\begin{proof}
[of Theorem \ref{Main}]The key idea of the proof is that the function
\begin{equation}
V\left(  x\right)  =\left\vert f\left(  x\right)  -g\left(  x\right)
\right\vert \label{MyV}%
\end{equation}
is strictly monotonically decreasing along the trajectories of
(\ref{scalar DS}). From Lemma \ref{Main} we have:\vspace*{-2pt}%
\begin{align}
& V\left(  x\left(  t+1\right)  \right)  -V\left(  x\left(  t\right)  \right)
\nonumber\\
& =\left\vert f\left(  x\left(  t+1\right)  \right)  -g\left(  x\left(
t+1\right)  \right)  \right\vert -\left\vert f\left(  x\left(  t\right)
\right)  -g\left(  x\left(  t\right)  \right)  \right\vert \nonumber\\
& =\left\vert f\left(  x\left(  t+1\right)  \right)  -f\left(  x\left(
t\right)  \right)  \right\vert -\left\vert g\left(  x\left(  t+1\right)
\right)  -g\left(  x\left(  t\right)  \right)  \right\vert \nonumber\\
& <0.\label{LdanZ}%
\end{align}
Therefore, $\left\{  V\left(  x\left(  t\right)  \right)  \right\}  $ is a
strictly decreasing bounded sequence and converges to a limit $c\geq0$. We
show that $c>0$ is not possible. Note that the sequence $\left\{  x\left(
t\right)  \right\}  $ is bounded from below since the domain of $\psi$ is
$\mathbb{R}_{+}$. Furthermore, as long as $f\left(  x\left(  t\right)
\right)  <g\left(  x\left(  t\right)  \right)  ,$ the sequence $\left\{
g\left(  x\left(  t\right)  \right)  \right\}  $ decreases strictly.
Therefore, (\ref{fdlgd2}) implies that\vspace*{-2pt}%
\begin{equation}
\forall x_{0}:\exists~M\in\mathbb{R},~N\in\mathbb{Z}_{+}:g\left(  x\left(
t\right)  \right)  \leq M,\text{ }\forall t\geq N.\label{BonG}%
\end{equation}
It follows from (\ref{BonG}), monotonicity and continuity of $g\left(
\cdot\right)  $\ that the sequence $\left\{  x\left(  t\right)  \right\}  $ is
bounded from above too (similar arguments prove boundedness of $\left\{
x\left(  t\right)  \right\}  $ when (\ref{fdlgd3}) holds). Hence, either
$\lim_{t\rightarrow\infty}x\left(  t\right)  =0$,\ or $\left\{  x\left(
t\right)  \right\}  $ has a subsequence $\left\{  x\left(  t_{i}\right)
\right\}  $ which converges to a limit $x^{\ast}\in\mathbb{R}_{+}.$ In the
latter case we have\vspace*{-2pt}%
\begin{align*}
\lim_{t\rightarrow\infty}V\left(  x\left(  t\right)  \right)  \hspace
*{-0.09in}  & =\hspace*{-0.09in}\lim_{i\rightarrow\infty}V\left(  x\left(
t_{i}\right)  \right)  \hspace*{-0.02in}=\hspace*{-0.02in}\left\vert
\lim_{i\rightarrow\infty}\left\{  f\left(  x\left(  t_{i}\right)  \right)
-g\left(  x\left(  t_{i}\right)  \right)  \right\}  \right\vert \\[0.08in]
& =\left\vert f\left(  x^{\ast}\right)  -g\left(  x^{\ast}\right)
\right\vert
\end{align*}
\vspace*{10pt}If $g\left(  x^{\ast}\right)  =g\left(  \psi\left(  x^{\ast
}\right)  \right)  $ then $c=\left\vert f\left(  x^{\ast}\right)  -g\left(
\psi\left(  x^{\ast}\right)  \right)  \right\vert =0$ (due to (\ref{gp=f}))$.$
If $g\left(  x^{\ast}\right)  \neq g\left(  \psi\left(  x^{\ast}\right)
\right)  $ then
\[
\exists\delta,\varepsilon>0,\text{ s.t. }\left\vert g\left(  \psi\left(
x\right)  \right)  -g\left(  x\right)  \right\vert \geq\varepsilon,\text{
}\forall x\in\mathcal{B}\left(  x^{\ast},\delta\right)
\]
Define a function $\tau:\mathcal{B}\left(  x^{\ast},\delta\right)
\mapsto\overline{\mathbb{R}}_{+}$ according to%
\[
\tau:x\mapsto\frac{\left\vert f\left(  \psi\left(  x\right)  \right)
-f\left(  x\right)  \right\vert }{\left\vert g\left(  \psi\left(  x\right)
\right)  -g\left(  x\right)  \right\vert }%
\]
Then it follows from \ref{LdanZ} that $\tau\left(  x\right)  <1$ for all
$x\in\mathcal{B}\left(  x^{\ast},\delta\right)  $. Furthermore, the function
is continuous over the compact set $\mathcal{B}\left(  x^{\ast},\delta\right)
$ and achieves its supremum $\overline{\tau},$ where $\overline{\tau}<1$.
Since $x\left(  t_{i}\right)  $ converges to $x^{\ast},$ there exists
$\tilde{t}\in\mathbb{N},$ such that $x\left(  t\right)  \in\mathcal{B}\left(
x^{\ast},\delta\right)  .$ Then%
\begin{align*}
& V\left(  x\left(  t+1\right)  \right)  -\overline{\tau}V\left(  x\left(
t\right)  \right) \\[0.08in]
& =\left\vert f\left(  x\left(  t+1\right)  \right)  -f\left(  x\left(
t\right)  \right)  \right\vert -\overline{\tau}\left\vert g\left(  x\left(
t+1\right)  \right)  -g\left(  x\left(  t\right)  \right)  \right\vert
\\[0.08in]
& \leq0,\text{\quad}\forall t\geq\tilde{t}%
\end{align*}
Since $\overline{\tau}<1,$ this proves that $c=0.$ Finally,
\[
\lim_{t\rightarrow\infty}f\left(  x\left(  t\right)  \right)  =\lim
_{t\rightarrow\infty}g\left(  x\left(  t\right)  \right)  =g\left(  x^{\ast
}\right)
\]%
\[
x^{\ast}=g^{-1}(\lim_{t\rightarrow\infty}f\left(  x\left(  t\right)  \right)
)=\lim_{t\rightarrow\infty}g^{-1}\circ f\left(  x\left(  t\right)  \right)
=\lim_{t\rightarrow\infty}x\left(  t\right)
\]
This completes the proof of convergence for all initial conditions. Proof of
Lyapunov stability is based on standard arguments in proving stability of
nonlinear systems (see, e.g., \cite{Khalil}), while using the same Lyapunov
function defined in (\ref{MyV}).
\end{proof}

\begin{remark}
\label{Rem1}The monotonicity conditions in (\ref{fdlgd2}) or (\ref{fdlgd3}) in
Theorem \ref{Main}\ can be relaxed at the expense of more involved
technicalities in both the statement of the theorem and its proof. As we will
see, these assumptions are naturally satisfied in\ applications of interest to
this paper. Therefore, we will not bother with the technicalities of removing
the condition.
\end{remark}

There are situations in which a natural decomposition of discrete-time
dynamical systems\ via functions $f$ and $g$ satisfying (\ref{gp=f})\ is
readily available. This is often the case for applications that involve
optimization in a feedback loop, many instances of which appeared in section
\ref{DSDM}. For instance, for the price dynamics (\ref{Price D Simple}), we
have $\psi=\dot{c}\circ\dot{v}^{-1},$ and the decomposition is obtained with
$g=\dot{c}^{-1},$ and $f=\dot{v}^{-1}$. However, $f$ and $g$ obtained in this
way may not readily satisfy (\ref{fdlgd}). We present the following
corollaries.\medskip

\begin{corollary}
\label{CoCo}Consider the system (\ref{scalar DS}) and suppose that
continuously differentiable functions $f,g:\mathbb{R}_{+}\mapsto\mathbb{R}%
_{+}$ satisfying (\ref{gp=f}) and (\ref{fdlg4})$-$(\ref{fdlgd3}) are given.
Then, the system is stable if there exist $\theta\leq1$ and a strictly
monotonic, continuously differentiable\ function $\rho:\mathbb{R}_{+}%
\mapsto\mathbb{R}$ satisfying
\[
\left\vert \dot{\rho}\left(  f\left(  x\right)  \right)  \dot{f}\left(
x\right)  \right\vert \leq\theta\left\vert \dot{\rho}\left(  g\left(
x\right)  \right)  \dot{g}\left(  x\right)  \right\vert
\]
for all $x\in\mathbb{R}_{+}.$\vspace*{0.05in}
\end{corollary}

\begin{proof}
If $f$ and $g$ satisfy (\ref{gp=f}), then so do $\rho\circ f$ and $\rho\circ
g$ for any $\rho\in\mathcal{C}^{1}(0,\infty).$ Furthermore, under the
assumptions of the corollary, $\rho\circ f$ and $\rho\circ g$ satisfy
(\ref{fdlgd})$-$(\ref{fdlgd3}). The result then follows from Theorem
\ref{Main}.
\end{proof}

\begin{corollary}
\label{MarketSCol-I}\textbf{Market Stability I:} The system
(\ref{Price D Simple}) is stable if there exists a strictly monotonic,
continuously differentiable\ function $\rho:\mathbb{R}_{+}\mapsto\mathbb{R}$
satisfying%
\begin{equation}
\left\vert \dot{\rho}\left(  \dot{v}^{-1}\left(  \lambda\right)  \right)
\frac{\partial\dot{v}^{-1}\left(  \lambda\right)  }{\partial\lambda
}\right\vert \leq\theta\left\vert \dot{\rho}\left(  \dot{c}^{-1}\left(
\lambda\right)  \right)  \frac{\partial\dot{c}^{-1}\left(  \lambda\right)
}{\partial\lambda}\right\vert \label{rsslrvv}%
\end{equation}
for all $\lambda\in\mathbb{R}_{+}.\medskip$\newline Similarly, the system
(\ref{Demand D}) is stable if
\begin{equation}
\left\vert \dot{\rho}\left(  \dot{c}\left(  x\right)  \right)  \ddot{c}\left(
x\right)  \right\vert \leq\theta\left\vert \dot{\rho}\left(  \dot{v}\left(
x\right)  \right)  \ddot{v}\left(  x\right)  \right\vert \label{simple}%
\end{equation}
for all $x\in\mathbb{R}_{+}.$
\end{corollary}

\begin{proof}
The statements follow from Corollary \ref{CoCo} with $f=\dot{v}^{-1}$ and
$g=\dot{c}^{-1}$ for (\ref{rsslrvv}), and $f=\dot{c}$ and $g=\dot{v}$ for
(\ref{simple}), and the fact that under Assumption I, all of the conditions
required in Corollary \ref{CoCo} are satisfied.
\end{proof}

\bigskip

\begin{remark}
By taking $\rho\left(  \cdot\right)  $ to be the identity function in
(\ref{rsslrvv}) and (\ref{simple}), we obtain the following sufficient
criteria for stability of (\ref{Price D Simple}) or (\ref{Demand D}):%
\begin{equation}
\left\vert \ddot{c}\left(  x\right)  \right\vert \leq\theta\left\vert \ddot
{v}\left(  x\right)  \right\vert \label{cddotTvddot}%
\end{equation}
or%
\begin{equation}
\left\vert \frac{\partial\dot{v}^{-1}\left(  \lambda\right)  }{\partial
\lambda}\right\vert \leq\theta\left\vert \frac{\partial\dot{c}^{-1}\left(
\lambda\right)  }{\partial\lambda}\right\vert \label{vdotidotLTcdotidot}%
\end{equation}
Although these conditions are very simple, they are generally harder to
satisfy globally for typical supply and demand functions.
\end{remark}

In the economics literature, \textit{elasticity} is defined\ as a measure of
how one variable responds to a change in another variable. In particular,
\textit{price-elasticity of demand }is defined as the percentage change in the
quantity demanded, resulting from one percentage change in the price, and is
used as a measure of responsiveness, or sensitivity of demand to variations in
the price. \textit{Price-elasticity of supply} is defined analogously. In this
paper, we generalize the standard definitions of elasticity as follows.

\begin{definition}
\label{def:GenEl}\textbf{Generalized} \textbf{Elasticity:} The quantity%
\[
\epsilon_{D}^{\text{p}}\left(  \lambda,l\right)  =\left(  \frac{\lambda}%
{\dot{v}^{-1}\left(  \lambda\right)  }\right)  ^{l}\frac{\partial\dot{v}%
^{-1}\left(  \lambda\right)  }{\partial\lambda},\text{\quad}l\geq0
\]
is the generalized price-elasticity of demand at price $\lambda$. Similarly,%
\[
\epsilon_{S}^{\text{p}}\left(  \lambda,l\right)  =\left(  \frac{\lambda}%
{\dot{c}^{-1}\left(  \lambda\right)  }\right)  ^{l}\frac{\partial\dot{c}%
^{-1}\left(  \lambda\right)  }{\partial\lambda},\text{\quad}l\geq0
\]
is the generalized price-elasticity of supply at price $\lambda$. Note that
these notions depend on the exponent $l.$ For $l=1,$ we obtain the standard
notions of elasticity. We define the \textit{market's relative generalized
price-elasticity} as the ratio of the generalized price-elasticities:
\begin{equation}
\epsilon_{\text{rel}}^{\text{p}}\left(  \lambda,l\right)  =\frac{\epsilon
_{D}^{\text{p}}\left(  \lambda,l\right)  }{\epsilon_{S}^{\text{p}}\left(
\lambda,l\right)  }.\label{relPel}%
\end{equation}
The \textit{market's maximal relative price-elasticity} (MRPE) is defined as
\begin{equation}
\theta^{\ast}\left(  l\right)  =\sup_{\lambda\in\mathbb{R}_{+}}\left\vert
\epsilon_{\text{rel}}^{\text{p}}\left(  \lambda,l\right)  \right\vert
.\label{maxrelPel}%
\end{equation}
The notions of generalized demand-elasticity of price and generalized
supply-elasticity of price are defined analogously:%
\[
\epsilon_{p}^{\text{d}}\left(  x,l\right)  =x^{l}\frac{\ddot{v}\left(
x\right)  }{\dot{v}\left(  x\right)  ^{l}},\text{\quad}\epsilon_{p}^{\text{s}%
}\left(  x\right)  =x^{l}\frac{\ddot{c}\left(  x\right)  }{\dot{c}\left(
x\right)  ^{l}}%
\]
When $l=1,$ these notions coincide with the Arrow-Pratt coefficient of Risk
Aversion (RA) \cite{Arrow, Pratt}, and we will adopt the same jargon in this
paper. The market's \textit{relative generalized risk aversion factor} is
defined as:%
\[
\epsilon_{\text{rel}}^{\text{{\small RA}}}\left(  x,l\right)  =\frac
{\epsilon_{p}^{\text{s}}\left(  x,l\right)  }{\epsilon_{p}^{\text{d}}\left(
x,l\right)  }=\frac{\ddot{c}\left(  x\right)  }{\ddot{v}\left(  x\right)
}\left(  \frac{\dot{v}\left(  x\right)  }{\dot{c}\left(  x\right)  }\right)
^{l}%
\]
Finally, the \textit{market's maximal relative risk-aversion} (MRRA)\ is
defined as
\begin{equation}
\eta^{\ast}\left(  l\right)  =\sup\limits_{x\in\mathbb{R}_{+}}\,\left\vert
\epsilon_{\text{rel}}^{\text{{\small RA}}}\left(  x,l\right)  \right\vert
.\label{maxrelRA}%
\end{equation}
With a slight abuse of notation, when $l=1,$ we write $\epsilon_{D}^{\text{p}%
}\left(  \lambda\right)  $ instead of $\epsilon_{D}^{\text{p}}\left(
\lambda,1\right)  ,$ and $\theta^{\ast}$ instead of $\theta^{\ast}\left(
1\right)  $, etc.
\end{definition}

The following corollary relates the market's stability to the market's
relative price-elasticity $\epsilon_{\text{rel}}^{\text{p}}\left(
\lambda,l\right)  $, and relative risk-aversion $\epsilon_{\text{rel}%
}^{\text{{\small RA}}}\left(  x,l\right)  .$

\begin{corollary}
\label{MarketSCol-II}\textbf{Market Stability II:} The system (\ref{Price D})
is stable if the market's MRPE is less than one for some $l\geq0$, that is:%
\begin{equation}%
\begin{array}
[c]{c}%
\exists l\geq0:\quad\theta^{\ast}\left(  l\right)  =\sup\limits_{\lambda
}\,\left\vert \epsilon_{\text{rel}}^{\text{p}}\left(  \lambda,l\right)
\right\vert <1\medskip
\end{array}
\label{maxrelPelLd1}%
\end{equation}
The system (\ref{Demand D}) is stable if the market's MRRA is less than one
for some $k\geq0$, that is:%
\begin{equation}%
\begin{array}
[c]{c}%
\exists l\geq0:\quad\eta^{\ast}\left(  l\right)  =\sup\limits_{x}\,\left\vert
\epsilon_{\text{rel}}^{\text{{\small RA}}}\left(  x,l\right)  \right\vert
<1\medskip
\end{array}
\label{maxrelElLd1}%
\end{equation}

\end{corollary}

\begin{proof}
The results are obtained by applying Corollary \ref{MarketSCol-I}, criteria
(\ref{rsslrvv}) and (\ref{simple}),\ with $\rho\left(  z\right)  =\log\left(
z\right)  $ for $l=1,$ and $\rho\left(  z\right)  =z^{-l+1}$ for $l\neq1.$
\end{proof}

When the cost and value functions are explicitly available, conditions
(\ref{simple}) or (\ref{maxrelElLd1}) are more convenient to check, whereas,
when explicit expressions are available for the supply and demand functions,
it is more convenient to work with (\ref{rsslrvv}) or (\ref{maxrelPelLd1}).

\begin{example}
\label{EX1}Consider (\ref{Price D Simple}) with $c\left(  x\right)  =x^{\beta
},$ and $v\left(  x\right)  =\left(  x-u\right)  ^{1/\alpha},$ where
$\alpha,\beta>1$ and $u\geq0$ is a constant. First, consider the $u=0$ case.
Then, we have
\begin{gather*}
\lambda\left(  t+1\right)  =\beta\left(  \alpha\lambda\left(  t\right)
\right)  ^{\frac{\alpha\beta-\alpha}{1-\alpha}}\\[0.05in]
\dot{v}\left(  x\right)  =\alpha^{-1}x^{\frac{1-\alpha}{\alpha}},~\ddot
{v}\left(  x\right)  =\left(  1-\alpha\right)  \alpha^{-2}x^{\frac{1-2\alpha
}{\alpha}}\\[0.05in]
\dot{c}\left(  x\right)  =\beta x^{\beta-1},~\ddot{c}\left(  x\right)
=\beta\left(  \beta-1\right)  x^{\beta-2}%
\end{gather*}
It can be verified that there cannot exists a constant $\theta\in
\lbrack0,\infty)$ for which (\ref{cddotTvddot}) is satisfied for all
$x\in\overline{\mathbb{R}}_{+}$ (equivalently, $\theta^{\ast}\left(  0\right)
=\infty$). However$,$ by invoking (\ref{maxrelElLd1}) with $k=1,$ we have:%
\[
\eta^{\ast}=\frac{\left\vert \ddot{c}\left(  x\right)  \right\vert
}{\left\vert \ddot{v}\left(  x\right)  \right\vert }\frac{\left\vert \dot
{v}\left(  x\right)  \right\vert }{\left\vert \dot{c}\left(  x\right)
\right\vert }=\frac{\left(  \beta-1\right)  }{\left(  \alpha-1\right)
\alpha^{-1}}<1
\]
Hence, the system is stable if
\[
\beta<2-\alpha^{-1}%
\]
It can be shown that the condition is also necessary and the system diverges
for $\beta>2-\alpha^{-1}.$ Moreover, invoking (\ref{maxrelPelLd1}) with $l=1$
yields exactly the same result, though, this need not be the case in general.
Consider now the same system with $\alpha=\beta=2$ and $u>0.$ Simulations show
that the system is not stable in the asymptotic sense for $u<1/4.$ The
following table summarizes the results of our analysis.%
\begin{gather*}
\text{Table I}\\%
\begin{tabular}
[c]{|r|c|c|c|}\hline
& $u=0.25$ & $u=0.3$ & $u=0.5$\\\hline
$\theta^{\ast}\left(  1\right)  =$ & $2$ & $2$ & $2$\\\hline
$\theta^{\ast}\left(  1.5\right)  =$ & $1$ & $0.872$ & $0.595$\\\hline
$\theta^{\ast}\left(  2\right)  =$ & $1.299$ & $0.988$ & $0.459$\\\hline
\end{tabular}
\end{gather*}
Thus, when $u=1/4,$ the system is at least marginally stable. Furthermore, the
above analysis highlights the importance of the notion of generalized
elasticity introduced earlier (cf. Definition \ref{def:GenEl}), as
$\theta^{\ast}\left(  1\right)  $ (which is associated with the traditional
notion of price elasticity) can be greater than one while the system is stable
and it's stability may be proven using the\ MRPE for some $l\geq0.$
\end{example}

The preceding analysis is based on applying the results of Theorem \ref{Main}
and Corollary \ref{MarketSCol-II} to systems of the form (\ref{Price D}) (or
(\ref{Demand D})), which correspond to the persistence prediction model,
whether it is demand prediction by the ISO in the exant\'{e} pricing case,\ or
price prediction by the consumers in the ex-post pricing case. In the next
section, we present a theorem that is applicable to analysis under the generic
prediction models (\ref{PricePredStep}) and (\ref{PredStep}).

\subsection{Invariance Analysis}

When functions of the form (\ref{PricePredStep}) or (\ref{PredStep}) are used
for prediction of price or demand, the underlying dynamical system is no
longer a scalar system. An immediate extension of Theorem \ref{Main} in its
full generality to the multidimensional case, while possible, raises further
complexities in both the proof and the application of the theorem. In what
follows we take the middle way: we present a theorem that exploits the
structure of the dynamical system that arises from autoregressive prediction
models to both make the extension possible and to simplify the analysis.

\begin{theorem}
\label{MD-Invariance}Let $x:\overline{\mathbb{Z}}_{+}\rightarrow\mathbb{R},$
be a real-valued sequence satisfying a state-space equation of the form:%
\begin{align}
\hspace{-0.18in}g\left(  x\left(  t+1\right)  \right)  \hspace{-0.08in}  &
=\hspace{-0.08in}f\left(  x\left(  t\right)  ,x\left(  t-1\right)
,\cdots,x\left(  t-n\right)  \right)  \hspace*{0.12in}\label{generalARModel}%
\\[0.05in]
\hspace{-0.18in}\left(  x\left(  0\right)  ,...,x\left(  n\right)  \right)
\hspace{-0.08in}  & \in\hspace{-0.08in}X_{0}\subset\mathbb{R}^{n+1},\nonumber
\end{align}
for some continuously differentiable function $f:\mathbb{R}^{n+1}%
\mapsto\mathbb{R},$ and a continuously differentiable monotonic function
$g:\mathbb{R\mapsto R}$ which satisfy%
\begin{equation}
\left\vert \frac{\partial}{\partial y_{k}}f\left(  y\right)  \right\vert
\leq\theta_{k}\left\vert \overset{}{\dot{g}}\left(  y_{k}\right)  \right\vert
,\quad\forall y\in\mathbb{R}^{n+1}\label{MDCondition}%
\end{equation}
where%
\begin{equation}%
{\displaystyle\sum\limits_{k=1}^{n}}
\theta_{k}\leq1\label{sumthetald1}%
\end{equation}
Then, there exists a constant $\gamma_{0}\geq0,$ which depends only on the
first $n+1$ initial states $x\left(  n\right)  ,...,x\left(  0\right)  $, such
that the set
\begin{equation}
\Omega_{0}\hspace*{-0.02in}=\hspace*{-0.02in}\left\{  \hspace*{-0.02in}%
x\in\mathbb{R~}|~\exists z\in\mathbb{R}^{n}:\left\vert g\left(  x\right)
-f\left(  x,z\right)  \right\vert \leq\gamma_{0}\hspace*{-0.02in}\right\}
\label{MDInvSet}%
\end{equation}
is invariant under (\ref{generalARModel}), i.e.,%
\[
x\left(  T-n\right)  ,...,x\left(  T\right)  \in\Omega_{0}\Rightarrow x\left(
t\right)  \in\Omega_{0},\text{ }\forall t>T
\]
Furthermore, when (\ref{sumthetald1}) holds with strict inequality, the
$g$-scaled IAV of $x$ is bounded from above:%
\begin{equation}
\mathcal{V}_{g}\left(  x\right)  =%
{\displaystyle\sum\limits_{t=1}^{\infty}}
\left\vert g\left(  x\left(  t+1\right)  \right)  -g\left(  x\left(  t\right)
\right)  \right\vert \leq\frac{\gamma_{0}}{1-%
{\displaystyle\sum\limits_{k=1}^{n}}
\theta_{k}}\label{MDvarupbound}%
\end{equation}

\end{theorem}

\begin{proof}
For simplicity and convenience in notation, we prove the theorem for the $n=1
$ case$.$ The proof for the general case is entirely analogous. Define the
function $V:\mathbb{R}^{2}\mapsto\mathbb{R}_{+}$ according to%
\begin{equation}
V\left(  x,z\right)  =\left\vert g\left(  x\right)  -f\left(  x,z\right)
\right\vert \label{Vxz}%
\end{equation}
Let
\begin{equation}
\gamma_{0}=V\left(  x\left(  1\right)  ,x\left(  0\right)  \right)
+\left\vert g\left(  x\left(  1\right)  \right)  -g\left(  x\left(  0\right)
\right)  \right\vert \label{gam0}%
\end{equation}
To prove that $\Omega_{0}$ is invariant under (\ref{generalARModel}), it is
sufficient to show that%
\begin{equation}
V\left(  x\left(  T+1\right)  ,x\left(  T\right)  \right)  \leq\gamma
_{0},\text{ }\forall T\in\mathbb{Z}_{+}\label{toshow}%
\end{equation}
To simplify the notation, define $\Delta f_{t}=f\left(  x\left(  t+1\right)
,x\left(  t\right)  \right)  -f\left(  x\left(  t\right)  ,x\left(
t-1\right)  \right)  ,$ and $\Delta g_{t}=g\left(  x\left(  t+1\right)
\right)  -g\left(  x\left(  t\right)  \right)  .$ We have:\vspace*{0.1in}%
\[
\hspace*{-1in}V\left(  x\left(  t+1\right)  ,x\left(  t\right)  \right)
-V\left(  x\left(  t\right)  ,x\left(  t-1\right)  \right)
\]
\vspace*{-0.2in}%
\begin{multline*}
=\left\vert g\left(  x\left(  t+1\right)  \right)  -f\left(  x\left(
t+1\right)  ,x\left(  t\right)  \right)  \right\vert \\
-\left\vert g\left(  x\left(  t\right)  \right)  -f\left(  x\left(  t\right)
,x\left(  t-1\right)  \right)  \right\vert
\end{multline*}
\vspace*{-0.2in}%
\begin{multline*}
=\left\vert f\left(  x\left(  t\right)  ,x\left(  t-1\right)  \right)
-f\left(  x\left(  t+1\right)  ,x\left(  t\right)  \right)  \right\vert \\
-\left\vert g\left(  x\left(  t\right)  \right)  -g\left(  x\left(
t+1\right)  \right)  \right\vert
\end{multline*}
\vspace*{-0.2in}%
\begin{multline*}
\leq\left\vert f\left(  x\left(  t\right)  ,x\left(  t-1\right)  \right)
-f\left(  x\left(  t\right)  ,x\left(  t\right)  \right)  \right\vert \\
+\left\vert f\left(  x\left(  t\right)  ,x\left(  t\right)  \right)  -f\left(
x\left(  t+1\right)  ,x\left(  t\right)  \right)  \right\vert -\left\vert
\Delta g_{t}\right\vert
\end{multline*}%
\begin{equation}
\hspace*{-1.58in}\leq\theta_{2}\left\vert \Delta g_{t-1}\right\vert +\left(
\theta_{1}-1\right)  \left\vert \Delta g_{t}\right\vert \label{FourthIneq4_2}%
\end{equation}
where the first inequality is obtained by applying the triangular inequality,
and (\ref{FourthIneq4_2}) follows from (\ref{MDCondition}) and Lemma
\ref{MainLemma}. By summing up both sides of (\ref{FourthIneq4_2}) from $t=0$
to $t=T$ we obtain:%
\begin{multline}
V\left(  x\left(  T+1\right)  ,x\left(  T\right)  \right)  \leq V\left(
x\left(  1\right)  ,x\left(  0\right)  \right) \label{VTldanV0}\\
+\left(  \theta_{1}+\theta_{2}-1\right)
{\displaystyle\sum\limits_{t=1}^{T}}
\left\vert \Delta g_{t}\right\vert +\theta_{2}\left(  \left\vert \Delta
g_{0}\right\vert -\left\vert \Delta g_{T}\right\vert \right)
\end{multline}
The inequality (\ref{toshow})\ then follows from (\ref{VTldanV0}) and
(\ref{sumthetald1}). When (\ref{sumthetald1}) holds with strict inequality,
(\ref{MDvarupbound}) follows from (\ref{VTldanV0}) and nonnegativity of
$V\left(  x\left(  T+1\right)  ,x\left(  T\right)  \right)  $ for all
$T\in\mathbb{Z}.$ This completes the proof.
\end{proof}

It follows from the proof of Theorem \ref{MD-Invariance} that when the initial
conditions are close to the equilibrium of (\ref{generalARModel}), it is
sufficient to satisfy conditions (\ref{MDCondition})--(\ref{sumthetald1}) only
locally, over a properly defined subset of $\mathbb{R}^{n+1}$. This is
summarized in the following corollary.

\begin{corollary}
\label{col:MD-Invariance}Let $x:\overline{\mathbb{Z}}_{+}\rightarrow
\mathbb{R},$ be a real-valued sequence satisfying (\ref{generalARModel}),
where $f$ and $g$ are continuously differentiable functions. Let%
\[
\widetilde{\Omega}_{0}\hspace*{-0.02in}=\hspace*{-0.02in}\left\{
\hspace*{-0.02in}\left(  x,z\right)  \in\mathbb{R}\times\mathbb{R}%
^{n}:\left\vert g\left(  x\right)  -f\left(  x,z\right)  \right\vert
\leq\gamma_{0}\hspace*{-0.02in}\right\}
\]
where $\gamma_{0}$ is given in (\ref{Vxz})--(\ref{gam0}). If
\[
\left\vert \frac{\partial}{\partial y_{k}}f\left(  y\right)  \right\vert
\leq\theta_{k}\left\vert \overset{}{\dot{g}}\left(  y_{k}\right)  \right\vert
,\quad\forall y\in\widetilde{\Omega}_{0}%
\]
where $\theta_{k}$'s satisfy (\ref{sumthetald1}), then $\widetilde{\Omega}%
_{0}$ is invariant under (\ref{generalARModel}). Furthermore, when
(\ref{sumthetald1}) holds with strict inequality, and the initialization
vector $x_{0}=\left[  x\left(  n\right)  ,\ldots,x\left(  0\right)  \right]
\,$is an element of $\widetilde{\Omega}_{0}$,$\ $then (\ref{MDvarupbound}) holds.
\end{corollary}

Theorem \ref{MD-Invariance} and Corollary \ref{col:MD-Invariance} can be
applied to analysis of market dynamics under the generic autoregressive
prediction models that were presented in Section \ref{DSDM}. This includes the
generic dynamical system models that were developed for price dynamics under
exant\'{e} or ex-post pricing ((\ref{MinCost})--(\ref{OptDemand}) and
(\ref{ExPCmplxPriceD})), as well as the aggregate demand dynamical systems
(\ref{TheEq3}) and (\ref{TheEq4}). The sets $\Omega_{0}$ or $\widetilde
{\Omega}_{0}$ being invariant implies that the difference between the
predicted demand and the actual supply (possibly scaled by some monotonic
function, e.g., $\log\left(  \cdot\right)  $) remains bounded.

\subsubsection{Analysis of Market Dynamics under Generic Autoregressive
Prediction Models}

In this section we examine the impact of linear autoregressive prediction
models on market stability. Consider the model (\ref{AX Price D Simple}),
repeated here for convenience:%
\[
\dot{c}^{-1}\left(  \lambda\left(  t+1\right)  \right)  =\sum\nolimits_{k=0}%
^{n}\alpha_{k}\dot{v}^{-1}\left(  \lambda\left(  t-k\right)  \right)
\]
We apply Theorem \ref{MD-Invariance}\ (alternatively Corollary
\ref{col:MD-Invariance}) with
\begin{equation}
g\left(  \lambda\right)  =\rho\left(  \dot{c}^{-1}\left(  \lambda\right)
\right) \label{rho1}%
\end{equation}
and%
\begin{equation}
f\left(  \lambda_{t},\ldots,\lambda_{t-n}\right)  =\rho\left(  \sum
\nolimits_{k=0}^{n}\alpha_{k}\dot{v}^{-1}\left(  \lambda_{t-k}\right)  \right)
\label{rho2}%
\end{equation}
We examine (\ref{rho1})$-$(\ref{rho2}) with $\rho\left(  z\right)
=\log\left(  z\right)  $ and $\rho\left(  z\right)  =z^{-l+1},$ $l\neq1.$
Conditions (\ref{MDCondition})$-$(\ref{sumthetald1}) then imply that the
following conditions are sufficient\ (for some $k\geq0$):%
\begin{align}
\left\vert \frac{\alpha_{k}\displaystyle\left.  \frac{\partial\dot{v}%
^{-1}\left(  \lambda\right)  }{\partial\lambda}\right\vert _{\lambda
=\lambda_{t-k}}}{\displaystyle\left[  \sum\nolimits_{j=0}^{n}\alpha_{j}\dot
{v}^{-1}\left(  \lambda_{t-j}\right)  \right]  ^{l}}\right\vert  & \leq
\theta_{k}\left\vert \epsilon_{S}^{\text{p}}\left(  \lambda_{t-k},l\right)
\right\vert \label{2complicated}\\%
{\displaystyle\sum\limits_{k=1}^{n}}
\theta_{k}  & \leq1\label{2complicatedb}%
\end{align}
Conditions (\ref{2complicated})--(\ref{2complicatedb}) are complicated and in
general demand numerical computation for verification. However, examination of
(\ref{2complicated}) near equilibrium is informative. Suppose that
(\ref{AX Price D Simple}) converges to an equilibrium price $\bar{\lambda}$.
Letting $\lambda_{t}=\lambda_{t-1}=\cdots=\lambda_{t-n}=\bar{\lambda},$ we
observe that the following condition is implied by (\ref{2complicated}%
)--(\ref{2complicatedb}):%
\begin{equation}
\exists l\geq0:\left\vert
{\displaystyle\sum\nolimits_{k=1}^{n}}
a_{k}\right\vert \left\vert \epsilon_{D}^{\text{p}}\left(  \bar{\lambda
},l\right)  \right\vert \leq\left\vert \epsilon_{S}^{\text{p}}\left(
\bar{\lambda},l\right)  \right\vert \left\vert
{\displaystyle\sum\nolimits_{k=1}^{n}}
a_{k}\right\vert ^{l}\label{localanalysis}%
\end{equation}
where $\epsilon_{D}^{\text{p}}\left(  \bar{\lambda},l\right)  $ and
$\epsilon_{D}^{\text{p}}\left(  \bar{\lambda},l\right)  $ are generalized
elasticities as defined in Definition \ref{def:GenEl}, evaluated at the
equilibrium. It can be shown that (\ref{localanalysis}) is equivalent to
$\epsilon_{\text{rel}}^{\text{p}}\left(  \bar{\lambda},1\right)  \leq1,$
independently of $l.$ Furthermore, for a large class of cost and value
functions, namely power functions of the form $c\left(  x\right)  =x^{\beta}$
and $v\left(  x\right)  =x^{1/\alpha},$ $\alpha,\beta\geq1,$ the equilibrium
relative elasticity $\theta\left(  \bar{\lambda}\right)  =\epsilon
_{\text{rel}}^{\text{p}}\left(  \bar{\lambda},1\right)  $ is independent of
the autoregressive coefficients $a_{k},$ $k=1,..,n.$ Thus, if the closed-loop
market is unstable under the persistent prediction model ($a_{1}=1,$
$a_{k}=0,$ $k\neq1$), then global stability cannot be verified for any linear
auto-regressive model of the form (\ref{AX Price D Simple}) using
(\ref{2complicated})--(\ref{2complicatedb}). Although this analysis is based
on sufficient criteria, it suggests that it may be difficult to globally
stabilize these systems via linear autoregressive prediction. Indeed,
extensive simulations show that such models will not globally
stabilize\textbf{\ }an unstable market, unless the MRPE\ is very close to one.
For values of $\theta^{\ast}>1.05$ global stabilization could not be achieved
in our simulations. Local stabilization is, however, possible for moderate
values of $\theta^{\ast},$ namely, $\theta^{\ast}\lessapprox3.$\medskip

\subsubsection{Analysis of Dynamics of Markets with Exogenous Inputs\medskip}

\begin{theorem}
\label{1D-Invariancewithu}Let $x:\overline{\mathbb{Z}}_{+}\rightarrow
\mathbb{R}$ and $u:\overline{\mathbb{Z}}_{+}\rightarrow\mathbb{R}$ be
real-valued sequences which satisfy a state-space equation of the form:%
\begin{align}
\hspace{-0.18in}g\left(  x\left(  t+1\right)  \right)  \hspace{-0.07in}  &
=\hspace{-0.07in}f\left(  x\left(  t\right)  ,u\left(  t\right)  \right)
,\quad u\left(  t\right)  \in U\label{dswithu}\\[0.05in]
x\left(  0\right)  \hspace{-0.07in}  & \in\hspace{-0.07in}X_{0}\subset
\mathbb{R}\nonumber
\end{align}
for some continuously differentiable function $f:\mathbb{R}^{2}\rightarrow
\mathbb{R}$ and a continuously differentiable monotonic function
$g:\mathbb{R}\rightarrow\mathbb{R}$ satisfying
\begin{equation}
\left\vert \frac{\partial}{\partial u}f\left(  x,u\right)  \right\vert
\leq1,\text{\quad}\forall x\in\mathbb{R},\text{ }u\in U\label{UConstraint}%
\end{equation}
and
\begin{equation}
\left\vert \frac{\partial}{\partial x}f\left(  x,u\right)  \right\vert
\leq\theta\left\vert \overset{}{\dot{g}}\left(  x\right)  \right\vert
,\text{\quad}\forall x\in\mathbb{R},\text{ }u\in U,\label{ThetaConstraint}%
\end{equation}
where
\[
U=\left\{  u\in\mathbb{R}:\left\vert u\right\vert \leq\kappa\right\}
\]
and $\kappa\in\left(  0,\infty\right)  ,$ and $\theta\in\lbrack0,1).$ Define%
\begin{equation}
\zeta_{\kappa}\left(  \theta\right)  =\kappa\frac{1+\theta}{1-\theta
}.\label{alphatheta}%
\end{equation}
Then, the set
\begin{equation}
\Omega\left(  \theta\right)  \hspace*{-0.02in}=\hspace*{-0.02in}\left\{
\hspace*{-0.02in}x:\left\vert
\begin{array}
[c]{c}%
\hspace{-0.3in}\vspace*{0.03in}%
\end{array}
\left\vert f\left(  x,\nu\right)  -g\left(  x\right)  \right\vert -\left\vert
\nu\right\vert
\begin{array}
[c]{c}%
\hspace{-0.3in}\vspace*{0.03in}%
\end{array}
\right\vert \leq\zeta_{\kappa}\left(  \theta\right)  ,\text{ }\forall\nu\in
U\hspace*{-0.02in}\right\} \label{InvSet}%
\end{equation}
is invariant under (\ref{dswithu}). Furthermore, the $g$-scaled IMV of $x$ is
bounded from above:%
\begin{equation}
\overline{\mathcal{V}}_{g}\left(  x\right)  =\lim_{T\rightarrow\infty}\frac
{1}{T}%
{\displaystyle\sum\limits_{t=1}^{T}}
\left\vert g\left(  x\left(  t+1\right)  \right)  -g\left(  x\left(  t\right)
\right)  \right\vert \leq\frac{2\kappa}{1-\theta}\label{1Dvarupbound}%
\end{equation}

\end{theorem}

\medskip

\begin{proof}
Define%
\[
V\left(  x\right)  =\sup\limits_{\nu\in U}\left\{  \left\vert
\begin{array}
[c]{c}%
\hspace{-0.3in}\vspace*{0.03in}%
\end{array}
\left\vert f\left(  x,\nu\right)  -g\left(  x\right)  \right\vert -\left\vert
\nu\right\vert
\begin{array}
[c]{c}%
\hspace{-0.3in}\vspace*{0.03in}%
\end{array}
\right\vert \right\}  -\zeta_{\kappa}\left(  \theta\right)  .
\]
It is sufficient to show that there exists $\tau\geq0,$ such that:%
\[
V\left(  x\left(  t+1\right)  \right)  -\tau V\left(  x\left(  t\right)
\right)  \leq0,\text{ }\forall t\in\mathbb{Z}_{+}.
\]
To simplify the notation, define $\Delta f_{t}=f\left(  x\left(  t+1\right)
,u\left(  t+1\right)  \right)  -f\left(  x\left(  t\right)  ,u\left(
t\right)  \right)  ,$ and $\Delta g_{t}=g\left(  x\left(  t+1\right)  \right)
-g\left(  x\left(  t\right)  \right)  .$ Then%
\begin{multline*}
V\left(  x\left(  t+1\right)  \right)  -\tau V\left(  x\left(  t\right)
\right) \\
=\sup\limits_{\nu\in U}\left\{  \left\vert
\begin{array}
[c]{c}%
\hspace{-0.3in}\vspace*{0.03in}%
\end{array}
\left\vert f\left(  x\left(  t+1\right)  ,\nu\right)  -g\left(  x\left(
t+1\right)  \right)  \right\vert -\left\vert \nu\right\vert
\begin{array}
[c]{c}%
\hspace{-0.3in}\vspace*{0.03in}%
\end{array}
\right\vert \right\} \\
-\tau\sup\limits_{\nu\in U}\left\{  \left\vert
\begin{array}
[c]{c}%
\hspace{-0.3in}\vspace*{0.03in}%
\end{array}
\left\vert f\left(  x\left(  t\right)  ,\nu\right)  -g\left(  x\left(
t\right)  \right)  \right\vert -\left\vert \nu\right\vert
\begin{array}
[c]{c}%
\hspace{-0.3in}\vspace*{0.03in}%
\end{array}
\right\vert \right\} \\
+\zeta_{\kappa}\left(  \theta\right)  \left(  \tau-1\right)
\end{multline*}%
\begin{multline}
\leq\sup\limits_{\nu\in U}\left\{  \left\vert
\begin{array}
[c]{c}%
\hspace{-0.3in}\vspace*{0.03in}%
\end{array}
\left\vert f\left(  x\left(  t+1\right)  ,\nu\right)  -f\left(  x\left(
t\right)  ,u\left(  t\right)  \right)  \right\vert -\left\vert \nu\right\vert
\begin{array}
[c]{c}%
\hspace{-0.3in}\vspace*{0.03in}%
\end{array}
\right\vert \right\} \label{FirstIneq1}\\
-\tau\left\vert \Delta g_{t}\right\vert +\tau\kappa+\zeta_{\kappa}\left(
\theta\right)  \left(  \tau-1\right)
\end{multline}%
\begin{multline}
\leq\sup\limits_{\nu\in U}\left\vert f\left(  x\left(  t+1\right)
,\nu\right)  -f\left(  x\left(  t\right)  ,\nu\right)  \right\vert
\label{SecondIneq2}\\
+\sup\limits_{\nu\in U}\left\{  \left\vert
\begin{array}
[c]{c}%
\hspace{-0.3in}\vspace*{0.03in}%
\end{array}
\left\vert f\left(  x\left(  t\right)  ,\nu\right)  -f\left(  x\left(
t\right)  ,u\left(  t\right)  \right)  \right\vert -\left\vert \nu\right\vert
\begin{array}
[c]{c}%
\hspace{-0.3in}\vspace*{0.03in}%
\end{array}
\right\vert \right\} \\
-\tau\left\vert \Delta g_{t}\right\vert +\tau\kappa+\zeta_{\kappa}\left(
\theta\right)  \left(  \tau-1\right)
\end{multline}%
\begin{equation}
\hspace{-0.84in}\leq\left(  \theta-\tau\right)  \left\vert \Delta
g_{t}\right\vert +\left(  1+\tau\right)  \kappa+\zeta_{\kappa}\left(
\theta\right)  \left(  \tau-1\right) \label{FourthIneq4}%
\end{equation}
where (\ref{FirstIneq1}) follows from the choice of $\nu=u\left(  t\right)  $
and $\left\vert u\left(  t\right)  \right\vert \leq\kappa,$ (\ref{SecondIneq2}%
) follows from the triangular inequality, and (\ref{FourthIneq4}) follows from
(\ref{UConstraint})--(\ref{ThetaConstraint}) and Lemma \ref{MainLemma}. The
desired result follows from the fact that the right-hand side of
(\ref{FourthIneq4}) will be non-positive for $\tau=\theta,$ and $\zeta
_{\kappa}\left(  \theta\right)  $ defined in (\ref{alphatheta}). To prove
(\ref{1Dvarupbound}), let $\tau=1$ in (\ref{FourthIneq4}) to obtain%
\begin{equation}
V\left(  x\left(  t+1\right)  \right)  -V\left(  x\left(  t\right)  \right)
\leq\left(  \theta-1\right)  \left\vert \Delta g_{t}\right\vert +2\kappa
\label{Ineq123}%
\end{equation}
Summing both sides of (\ref{Ineq123}) over all $t=0,1,..,T$ results in:%
\begin{equation}
V\left(  x\left(  T+1\right)  \right)  \leq V\left(  x\left(  0\right)
\right)  +\left(  \theta-1\right)
{\displaystyle\sum\limits_{t=1}^{T}}
\left\vert \Delta g_{t}\right\vert +2T\kappa\label{Ineq124}%
\end{equation}
It follows from (\ref{Ineq124})\ and non-negativity of $V\left(  x\left(
T+1\right)  \right)  +\zeta_{\kappa}\left(  \theta\right)  $ that%
\begin{equation}
\left(  1-\theta\right)
{\displaystyle\sum\limits_{t=1}^{T}}
\left\vert \Delta g_{t}\right\vert \leq2T\kappa+V\left(  x\left(  0\right)
\right)  +\zeta_{\kappa}\left(  \theta\right)  .\label{Ineq125}%
\end{equation}
The desired result (\ref{1Dvarupbound}) then follows immediately from
(\ref{Ineq124}) by dividing by $T$ and taking the limit as $T\rightarrow
\infty.$
\end{proof}

The following corollary is a local variant of Theorem \ref{1D-Invariancewithu}%
, and is useful for scenarios in which, there exists no positive number
$\theta<1$ such that (\ref{ThetaConstraint}) is satisfied for all
$x\in\mathbb{R},$ whereas it might be possible to satisfy the inequality
locally over a subset that contains $\Omega\left(  \theta\right)  $.

\begin{corollary}
\label{LocalInvariance}Let $x:\overline{\mathbb{Z}}_{+}\rightarrow\mathbb{R} $
and $u:\overline{\mathbb{Z}}_{+}\rightarrow\mathbb{R}$ be real-valued
sequences satisfying (\ref{dswithu}). For $\theta<1,$ define:%
\begin{multline*}
\widetilde{\theta}^{\ast}=\inf\left\{  \widetilde{\theta}:\left\vert
\frac{\partial}{\partial x}f\left(  x,u\right)  \right\vert \leq
\widetilde{\theta}\left\vert \frac{\partial}{\partial x}g\left(  x\right)
\right\vert \right.  ,\\
\left.
\begin{array}
[c]{c}%
\ \\
\
\end{array}
\text{\ }\forall x\in\Omega\left(  \theta\right)  ,\text{ }u\in U\right\}
\end{multline*}
where $\Omega\left(  \theta\right)  $ is given in (\ref{InvSet}). Then
$\Omega\left(  \widetilde{\theta}^{\ast}\right)  $ is invariant under
(\ref{dswithu}) if
\[
\widetilde{\theta}^{\ast}\leq\theta.
\]
Furthermore, (\ref{1Dvarupbound}) holds with $\theta=\widetilde{\theta}^{\ast
}.$
\end{corollary}

Consider equation (\ref{dswithu}) or (\ref{generalARModel}). When the
functions $g$ and $f$\ are $\rho$-scaled supply and demand functions, the
minimal $\theta$ satisfying (\ref{ThetaConstraint}) or (\ref{MDCondition})
will be the MRPE associated with the market models (\ref{dswithu}) or
(\ref{generalARModel}). When $g$ and $f$\ are $\rho$-scaled marginal value and
marginal cost functions respectively, the minimal $\theta$ satisfying the
inequalities will be the MRRA associated with (\ref{dswithu}) or
(\ref{generalARModel}). In the remainder of this section, we consider
applications of Theorem \ref{1D-Invariancewithu} to the two time-varying
models of consumer behavior (\ref{TVDF1a}) and (\ref{TVDF2a}).

\paragraph{Multiplicative Perturbation}

Consider the multiplicative perturbation model (\ref{TVDF2a}). Under this
model, the market dynamics is given by
\begin{equation}
\dot{c}^{-1}\left(  \lambda\left(  t+1\right)  \right)  =\left(  1+\frac{1}%
{2}\delta\left(  t\right)  \right)  \dot{v}^{-1}\left(  \lambda\left(
t\right)  \right)  ,\text{\quad}\delta\left(  t\right)  \in\left[
-\kappa,\kappa\right] \label{MultPurtDyn}%
\end{equation}
where the $1/2$ factor in front of $\delta\left(  t\right)  $ is simply a
scaling factor. We invoke Theorem \ref{1D-Invariancewithu} with
\begin{equation}
g\left(  \lambda\right)  =\log\left(  \dot{c}^{-1}\left(  \lambda\right)
\right) \label{glislogcdil}%
\end{equation}
and%
\begin{align*}
f\left(  \lambda,\delta\right)   & =\log\left(  1+\delta/2\right)  \dot
{v}^{-1}\left(  \lambda\right) \\
& =\log\left(  1+\delta/2\right)  +\log\left(  \dot{v}^{-1}\left(
\lambda\right)  \right)  .
\end{align*}
It can be verified that (\ref{UConstraint}) and (\ref{ThetaConstraint}) are
satisfied as long as $\kappa\leq1$ and $\theta^{\ast}<1,$ where $\theta^{\ast
}$ is the MRPE defined in (\ref{maxrelPel}). Furthermore, $\zeta_{\kappa
}\left(  \theta^{\ast}\right)  $ is the upperbound on the size of the
invariant set, where $\zeta_{\kappa}\left(  \cdot\right)  $ is defined in
(\ref{alphatheta}). In particular as $\theta^{\ast}\rightarrow1,$ small
perturbations may induce extremely large fluctuations as measured by $\log
$-scaled IMV of supply. The theoretical upperbound is $1/\left(
1-\theta^{\ast}\right)  .$ When Corollary \ref{LocalInvariance} is applicable,
the size of the invariant set can be characterized by $\zeta_{\kappa
}(\widetilde{\theta}^{\ast}),$ where $\widetilde{\theta}^{\ast}$ is the
market's local relative price-elasticity. Furthermore, volatility can be
characterized by $\widetilde{\theta}^{\ast}$ as well.

\paragraph{Additive Perturbation}

Under the additive perturbation model (\ref{TVDF1a}), the market dynamics can
be written as%
\begin{equation}
\dot{c}^{-1}\left(  \lambda\left(  t+1\right)  \right)  =u_{0}+\frac{1}%
{2}u\left(  t\right)  +\dot{v}^{-1}\left(  \lambda\left(  t\right)  \right)
,\text{\quad}u\left(  t\right)  \in\left[  -\kappa,\kappa\right]
\label{PeriodicD}%
\end{equation}
where $u_{0}\geq1$ is a shifting factor, and $\kappa\leq u_{0},$ so that the
demand is always at least $u_{0}/2.$ We invoke Theorem
\ref{1D-Invariancewithu} with (\ref{glislogcdil}) and%
\[
f\left(  \lambda,u\right)  =\log\left(  u_{0}+\frac{1}{2}u+\dot{v}^{-1}\left(
\lambda\right)  \right)
\]
Then, under the above assumptions, (\ref{UConstraint}) is satisfied. In a
similar fashion to previous analyzes, (\ref{ThetaConstraint}) can be related
to the market's relative price-elasticity. In this case, the price-elasticity
of demand turns out to be:
\[
\epsilon_{D}\left(  \lambda\right)  =\frac{\partial f\left(  \lambda,u\right)
}{\partial\lambda}=\frac{\lambda}{u_{0}+u/2+\dot{v}^{-1}\left(  \lambda
\right)  }\frac{\partial\dot{v}^{-1}\left(  \lambda\right)  }{\partial\lambda}%
\]
The larger the minimum of the inelastic component (i.e., $u_{0}-\kappa/2$),
the smaller the price-elasticity of the overall demand will be. Under the
assumptions made above, there is always a nonzero minimal demand $u_{\min
}\left(  t\right)  =u_{0}/2$. Therefore, it is sufficient to verify
(\ref{ThetaConstraint}) over $\lambda\geq\dot{c}\left(  u_{0}/2\right)  $
instead of all $\lambda>0.$ In conclusion, (\ref{ThetaConstraint}) reduces
to:\
\begin{equation}
\left\vert \frac{\partial\dot{v}^{-1}\left(  \lambda\right)  /\partial\lambda
}{u_{0}/2+\dot{v}^{-1}\left(  \lambda\right)  }\right\vert \leq\theta
\left\vert \frac{\partial\dot{c}^{-1}\left(  \lambda\right)  /\partial\lambda
}{\dot{c}^{-1}\left(  \lambda\right)  }\right\vert ,\quad\forall\lambda
\geq\dot{c}\left(  u_{0}/2\right) \label{dadadadada}%
\end{equation}
Let $\widetilde{\theta}^{\ast}$ be the minimal $\theta$ satisfying
(\ref{dadadadada}). Similar to the case with multiplicative uncertainty, in
this case too, the upperbound on the size of the invariant set is given by
$\zeta_{\kappa}(\widetilde{\theta}^{\ast})$, where $\zeta_{\kappa}\left(
\cdot\right)  $ is given in (\ref{alphatheta}). Moreover, the $\log$-scaled
IMV of supply is upperbounded by $u_{0}/(1-\widetilde{\theta}^{\ast}).$

The analysis\ reconfirms the intuition that participation of a small portion
of the population in real-time pricing will not have a severe destabilizing
effect on the system, as satisfying (\ref{dadadadada}) for larger values of
$u_{0}$ is easier. System stability concerns should arise when a large portion
of the population is exposed to real-time pricing.

\begin{remark}
It can be proven that when $u\left(  t\right)  $ is a periodic function with
period $T$ and (\ref{ThetaConstraint}) is satisfied, then all solutions of
(\ref{PeriodicD})\ converge to a periodic trajectory with period $T.$
\end{remark}

\subsection{Volatility}

The following corollaries follow from Theorems \ref{MD-Invariance}\ and
\ref{1D-Invariancewithu}, and explicitly relate the market's MRPE and MRRA to volatility.

\begin{corollary}
\label{Volatility I}\textbf{Volatility I:} Let $\theta^{\ast}<1$ and
$\eta^{\ast}<1$ be the MRPE and MRRA associated with the market model
(\ref{dswithu}). Then, there exists a constant $C$, depending on the size of
the disturbances only, such that the $\log$-scaled IMV of supply is
upperbounded by $C/\left(  1-\theta^{\ast}\right)  ,$ i.e.,%
\begin{equation}
\lim_{T\rightarrow\infty}\frac{1}{T}%
{\displaystyle\sum\limits_{t=1}^{T}}
\left\vert \log\left(  \dot{c}^{-1}\left(  \lambda\left(  t+1\right)  \right)
\right)  -\log\left(  \dot{c}^{-1}\left(  \lambda\left(  t\right)  \right)
\right)  \right\vert \leq\frac{C}{1-\theta^{\ast}}\label{vol-I}%
\end{equation}
And the $\log$-scaled IMV of price is upperbounded by $C/\left(  1-\eta^{\ast
}\right)  ,$ i.e.,%
\begin{equation}
\lim_{T\rightarrow\infty}\frac{1}{T}%
{\displaystyle\sum\limits_{t=1}^{T}}
\left\vert \log\left(  \lambda\left(  t+1\right)  \right)  -\log\left(
\lambda\left(  t\right)  \right)  \right\vert \leq\frac{C}{1-\eta^{\ast}%
}\medskip\label{vol-II}%
\end{equation}

\end{corollary}

\begin{corollary}
\label{Volatility II}\textbf{Volatility II:} Let $\theta^{\ast}<1$ and
$\eta^{\ast}<1$ be the MRPE and MRRA associated with the market model
(\ref{generalARModel}) with linear autoregressive prediction. Then, there
exists a constant $C$ such that the $\log$-scaled IAV of supply is
upperbounded by $C/\left(  1-\theta^{\ast}\right)  ,$ i.e.,%
\begin{equation}%
{\displaystyle\sum\limits_{t=1}^{\infty}}
\left\vert \log\left(  \dot{c}^{-1}\left(  \lambda\left(  t+1\right)  \right)
\right)  -\log\left(  \dot{c}^{-1}\left(  \lambda\left(  t\right)  \right)
\right)  \right\vert \leq\frac{C}{1-\theta^{\ast}}\label{vol-III}%
\end{equation}
And the $\log$-scaled IAV of price is upperbounded by $C/\left(  1-\eta^{\ast
}\right)  ,$ i.e.,%
\begin{equation}%
{\displaystyle\sum\limits_{t=1}^{\infty}}
\left\vert \log\left(  \lambda\left(  t+1\right)  \right)  -\log\left(
\lambda\left(  t\right)  \right)  \right\vert \leq\frac{C}{1-\eta^{\ast}%
}\medskip\label{vol-IV}%
\end{equation}

\end{corollary}

\begin{remark}
Generalized versions\ of the above corollaries can be formulated based on
$\theta^{\ast}\left(  l\right)  $ and $\eta^{\ast}\left(  l\right)  ,$ in
which case the scalings of the signals need to be defined\ accordingly:
letting $\rho_{l}\left(  x\right)  =x^{-l+1}$ for $l\neq1,$ the $\rho_{l}%
$-scaled IMV of supply and price will be upperbounded by $C/\left(
1-\theta^{\ast}\left(  l\right)  \right)  $ and $C/\left(  1-\eta^{\ast
}\left(  l\right)  \right)  $ respectively$.$ Furthermore, when the prices
remain bounded within an invariant set, e.g., when the conditions of Corollary
\ref{col:MD-Invariance} or Corollary \ref{LocalInvariance} hold, one can
replace $\theta^{\ast}\left(  l\right)  $ and $\eta^{\ast}\left(  l\right)  $
with local relative elasticity ratios $\widetilde{\theta}^{\ast}\left(
l\right)  $ and $\widetilde{\eta}^{\ast}\left(  l\right)  .$
\end{remark}

\subsection{Robustness and Incremental L2-Gain}

The $\rho$\textit{-scaled incremental L2-gain} of a discrete-time dynamical
system with input signal $u:\mathbb{Z}\rightarrow\mathbb{R}$ and output signal
$h:\mathbb{Z}\rightarrow\mathbb{R}$ is defined to be the minimal $\gamma\geq0$
such that the inequality
\begin{equation}
\gamma\left\Vert \rho\left(  u\right)  -\rho\left(  \bar{u}\right)
\right\Vert _{2}-\left\Vert \rho\left(  h\right)  -\rho\left(  \bar{h}\right)
\right\Vert _{2}\geq0\label{IncL2GDef}%
\end{equation}
is satisfied for all input/output pairs $\left(  u,h\right)  $ and $\left(
\bar{u},\bar{h}\right)  $ such that%
\[
\rho\left(  u\right)  -\rho\left(  \bar{u}\right)  \in\ell_{2}.
\]
For systems with larger gains, it is generally expected that
\textit{relatively small} deviations from a nominal input $\bar{u}$ would stir
\textit{relatively larger} deviations from the nominal output signal $\bar{h}
$. This gain can be used as a metric for assessing the robustness/sensitivity
of the system to arbitrary external disturbances. It can be proven that for
the market model (\ref{MultPurtDyn})\ (more generally, the market model
(\ref{dswithu}) with multiplicative uncertainty), the $\log$-scaled
incremental L2-gain from the perturbation $\delta\left(  \cdot\right)  $ to
the demand is upperbounded by $\theta^{\ast}/\left(  1-\theta^{\ast}\right)
.$ The gain from $\delta\left(  \cdot\right)  $ to the supply is upperbounded
by $1/\left(  1-\theta^{\ast}\right)  .$ These results--stated here without
proof--quantify the dependence of the closed-loop system's robustness, as
measured by the incremental L2-gain, on the markets maximal relative price-elasticity.

\section{Discussion\label{sec:dis}}

Cho and Meyn \cite{Meycho} have investigated the problem of volatility of
power markets in a dynamic general equilibrium framework. Their model can be
viewed as a full-information model in which the system operator has full
information about the cost and value functions of the producers and consumers.
Market clearing is instantaneous and supply and demand are matched with no
time lag. The producer's problem is, however, subject to supply friction or a
ramp constraint, i.e., a finite bound on the rate of change of supply
capacity. It is concluded that efficient equilibria are volatile and
volatility is attributed to the supply friction. In the formulation of
\cite{Meycho} the consumer's problem is not subject to ramp constraints. In
our formulation, neither the consumer's problem nor the producer's is
explicitly subject to ramp constraints, yet other factors are shown to
contribute to volatility, namely, information asymmetry and high price
elasticity of demand. Interestingly, if we included ramp constraints in the
consumer's problem it would have a stabilizing effect, as it would limit the
consumer's responsiveness to price signals and reduce her elasticity. This
effect is implicitly and qualitatively captured in our framework through the
introduction of an inelastic component in the demand, which certainly limits
the rate of change in the demand, and was shown to have a stabilizing effect.
However, uncertainty in the supply side, either in the available capacity or
in the cost, works in the reverse direction:\ when supply is sufficiently
volatile, a trade-off might exist and responsiveness and increased elasticity
of demand might be desirable, though this needs to be quantified rigorously.
The models developed in the paper do not include uncertainty in generation,
and this would be an interesting direction for future research.

The above discussion leads to another interesting question: "quantifying the
value of information in closed-loop electricity markets". Given the
heterogeneous nature of consumers and time-varying uncertainty in their
preferences, needs, and valuations for electricity, learning their value
functions and predicting their response to a price signal in real-time appears
to be a difficult problem. Suppose that the consumers provide a real-time
estimate of their inelastic and elastic consumption to the ISO. How valuable
will this real-time information be and what would be its impact on volatility
and reliability of the system? Given the potentially significant costs and
barriers associated with obtaining such information in real-time, quantifying
the value of information in this context seems an extremely important and
timely question with potentially significant impact the architecture of future
power grids.

\section{Numerical Simulations\label{sec:sim}}

In this section we present the results of some numerical simulation. For the
purpose of simulations, we use the following demand model:%
\begin{equation}
D\left(  t\right)  =\mu_{1}d_{1}\left(  t\right)  +\mu_{2}\left(  1+\delta
_{2}\left(  t\right)  \right)  \dot{v}^{-1}\left(  \lambda\left(  t\right)
\right) \label{simdemF}%
\end{equation}
where $d_{1}\left(  t\right)  $ is the exogenous, inelastic demand:%
\[
d_{1}\left(  t\right)  =a_{0}+a_{1}\sin\left(  t\right)  +a_{2}\sin\left(
2t\right)  +\delta_{1}\left(  t\right)
\]
and $\delta_{1}\left(  t\right)  \sim\mathcal{N}\left(  0,0.1^{2}\right)  $
and $\delta_{2}\left(  t\right)  \sim\mathcal{N}\left(  0,0.01^{2}\right)  $
are random disturbances. The parameters $\mu_{1}$ and $\mu_{2}$ are adjusted,
on a case-to-case basis, such that the average demand under real-time pricing
(i.e., when $\mu_{2}>0,$ $\mu_{1}<1$) remains nearly equal to the average
demand in the open loop market ($\mu_{2}=0,$ $\mu_{1}=1$), that is:%
\[%
{\displaystyle\sum\nolimits_{t=1}^{N}}
D\left(  t\right)  \approx%
{\displaystyle\sum\nolimits_{t=1}^{N}}
d_{1}\left(  t\right)
\]
This normalization, takes out the effect of higher or lower average demand on
price and allows for a fair comparison of volatility of prices in open-loop
and closed-loop markets. The following parameters are chosen for all
simulations in this section:%
\[
a_{0}=4\text{ GW, }a_{1}=1\text{ GW, }a_{2}=1\text{ GW}%
\]
This puts the peak of the inelastic demand at $6$ GW and the valley at $2$ GW,
modulo the random disturbance $\delta_{1}\left(  t\right)  .$ All simulations
are for a $24$ hour period and prices are updated every $5$ minutes. The
average demand in all simulations is approximately $4$ GW per five minutes for
both open-loop and closed-loop markets. The metric for comparison in these
simulation is the Relative Volatility Ratio (RVR), defined as the ratio of the
$\log$-scaled IAV of the closed-loop market to the $\log$-scaled IAV of the
open-loop market. The results of the first simulation are summarized in Figure
\ref{FIGURE2}. The prices are extremely volatile under real-time pricing (RVR
$=51.12$) and the system is practically unstable.%

\begin{figure}
[ptbh]
\begin{center}
\includegraphics[
height=2.4016in,
width=3.1661in
]%
{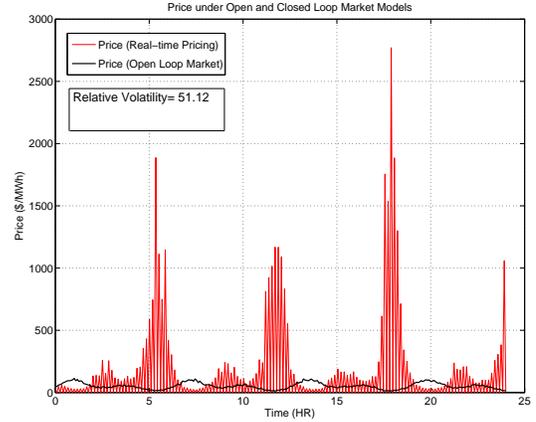}%
\caption{Simulation of a market with quadratic cost function\ $c\left(
x\right)  =x^{3},$ value function $v\left(  x\right)  =\log\left(  x\right)
,$ and demand function $D\left(  t\right)  $ given in (\ref{simdemF}) with
$\mu_{1}=0.075,$ $\mu_{2}=2.$}%
\label{FIGURE2}%
\end{center}
\end{figure}

The results of the second simulation are summarized in Figure \ref{FIGURE3}.
Based on the chosen parameters, this market is less volatile than the one in
the first simulation, yet, volatility of demand increases under real-time
pricing (RVR=$2.33$). Since in this simulation the cost is quadratic, the
price (not shown)\ has a very similar pattern.%

\begin{figure}
[ptbh]
\begin{center}
\includegraphics[
height=2.4016in,
width=3.1652in
]%
{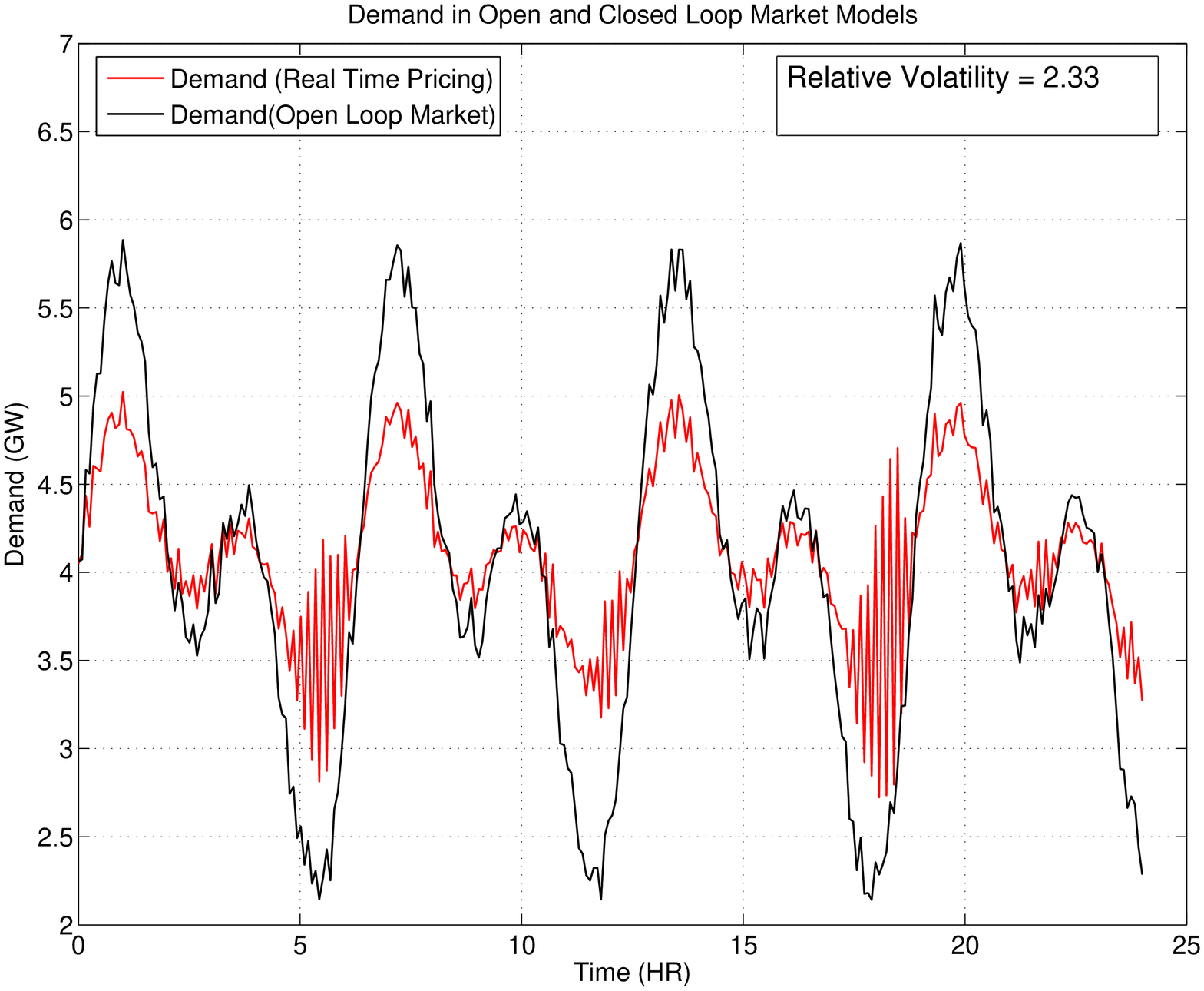}%
\caption{Simulation of a market with quadratic cost function\ $c\left(
x\right)  =3x^{2},$ value function $v\left(  x\right)  =\sqrt{x},$ and demand
function $D\left(  t\right)  $ in (\ref{simdemF}) with $\mu_{1}=0.7,$ $\mu
_{2}=3\times10^{3}.$ }%
\label{FIGURE3}%
\end{center}
\end{figure}

The third simulation is summarized in Figure \ref{FIGURE3}. For each value of
$\mu_{1}\in\left[  0,1\right]  $ (with $0.05$ increments), the expected RVR
was calculated by taking the average RVR of $50$ randomized simulations. The
random parameters are $\delta_{1}\left(  t\right)  ,$ $\delta_{2}\left(
t\right)  ,$ and the initial conditions. The experiment was repeated for four
different value functions: $v\left(  x\right)  =x^{1/a},$ $a=4,4.5,5,5.5.$ It
is observed that volatility increases with decreasing $a$ or $\mu_{1},$ both
of which increase the price-elasticity of demand.%

\begin{figure}
[ptbh]
\begin{center}
\includegraphics[
height=2.5598in,
width=3.3762in
]%
{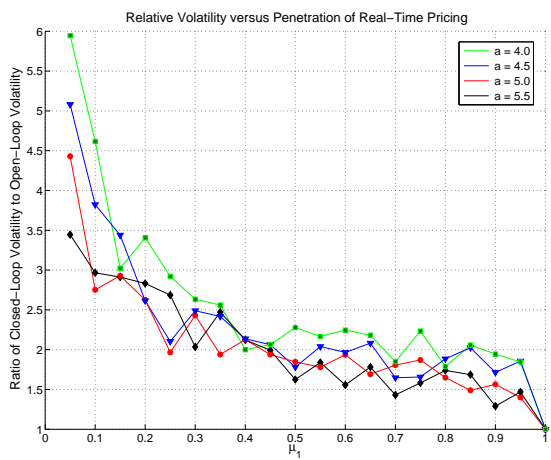}%
\caption{Simulation of a market with quadratic cost function $c(x)=3x^{2},$
value function $v(x)=x^{1/a}$, and demand function $D(t)$ given in
(\ref{simdemF}) with $\mu_{1}\in\lbrack0,1],$ and $\mu_{2}$ adjusted
accordingly to keep the total demand constant. Decreasing $a$ or $\mu_{1}$
increase the price-elasticity of the overall demand and hence, increase
volatility.}%
\label{FIGURE4}%
\end{center}
\end{figure}

\section{Conclusions and Future Work\label{sec:concl}}

We investigated the effects of real-time pricing on the stability and
volatility of electricity markets, and showed that exposing the retail
consumers to the real-time wholesale market prices creates a closed-loop
feedback system which could be very volatile or even unstable. When the system
is stable, an upper bound on volatility and robustness to external
disturbances can be characterized in terms of the market's relative
price-elasticity, defined as the ratio of \textit{generalized }%
price-elasticity of consumers to that of the producers. As this ratio
increases, the system may become more volatile, eventually becoming unstable
when the ratio exceeds one. As the penetration of new demand response
technologies and distributed storage within the power grid increases, so does
the price-elasticity of demand, and this is likely to increase volatility and
possibly destabilize the system under current market and system operation
practices. While the system can be stabilized and volatility can be reduced in
many different ways, e.g., via static or dynamic controllers regulating the
interaction of wholesale markets and retail consumers, different pricing
mechanisms pose different limitations on competing factors of interest. In
light of this, systematic analysis of the implications of different pricing
mechanisms, and quantifying the value of information and characterization of
the fundamental trade-offs between price volatility and economic efficiency,
as well as system reliability and environmental efficiency are important
directions of future research. In summary, more sophisticated models of
demand, a deeper understanding of consumer behavior in response to real-time
prices, and a thorough understanding of the implications of different market
mechanisms and system architectures are needed before real-time pricing can be
implemented in large-scale.

\end{document}